\documentclass[letterpaper,11pt]{article}
\usepackage{authblk}
\usepackage[margin=1in]{geometry}
\usepackage{makecell}
\usepackage{hyperref}
\usepackage{booktabs}
\usepackage{enumerate}
\usepackage{newunicodechar}
\usepackage{silence}
\newunicodechar{♢}{\tikz \node[inner sep=1.5,draw,diamond] {};}
\newunicodechar{☆}{\tikz \node[inner sep=1,draw,star,star point ratio=2] {};}
\newunicodechar{△}{\triangle}
\newunicodechar{⬜}{\kern 0.5pt\tikz \node[inner sep=1.7,draw,regular polygon,regular polygon sides=4] {};\kern 0.5pt}
\newunicodechar{◯}{\tikz[baseline=-3pt] \node[inner sep=1.7,draw,cloud,cloud puffs=4,cloud puff arc=190] {};}

\newunicodechar{⊥}{\bot}
\newunicodechar{•}{\item}
\newunicodechar{✓}{\checkmark}
\newunicodechar{✗}{\xmark}\def\xmark{\ding{55}}%
\newunicodechar{…}{\dots}
\newunicodechar{≔}{\coloneqq}
\newunicodechar{⁻}{^-}
\newunicodechar{⁺}{^+}
\newunicodechar{₋}{_-}
\newunicodechar{₊}{_+}
\newunicodechar{ℓ}{\ell}
\newunicodechar{•}{\item}
\newunicodechar{…}{\dots}
\newunicodechar{≔}{\coloneqq}
\newif\ifnormopen\normopenfalse
\newunicodechar{‖}{\ifnormopen\rVert\normopenfalse\else\rVert\normopentrue\fi}
\newunicodechar{≤}{\leq}
\newunicodechar{≥}{\geq}
\newunicodechar{≰}{\nleq}
\newunicodechar{≱}{\ngeq}
\newunicodechar{⊕}{\oplus}
\newunicodechar{⊗}{\otimes}
\newunicodechar{≠}{\neq}
\newunicodechar{¬}{\neg}
\newunicodechar{≡}{\equiv}
\newunicodechar{₀}{_0}
\newunicodechar{₁}{_1}
\newunicodechar{₂}{_2}
\newunicodechar{₃}{_3}
\newunicodechar{₄}{_4}
\newunicodechar{₅}{_5}
\newunicodechar{₆}{_6}
\newunicodechar{₇}{_7}
\newunicodechar{₈}{_8}
\newunicodechar{₉}{_9}
\newunicodechar{ₚ}{_p}
\newunicodechar{ₙ}{_n}
\newunicodechar{ₐ}{_a}
\newunicodechar{ₑ}{_e}
\newunicodechar{ₕ}{_h}
\newunicodechar{ₖ}{_k}
\newunicodechar{ₗ}{_l}
\newunicodechar{ₘ}{_m}
\newunicodechar{ₛ}{_s}
\newunicodechar{ₜ}{_t}
\newunicodechar{ₓ}{_x}
\newunicodechar{⁰}{^0}
\newunicodechar{¹}{^1}
\newunicodechar{²}{^2}
\newunicodechar{³}{^3}
\newunicodechar{⁴}{^4}
\newunicodechar{⁵}{^5}
\newunicodechar{⁶}{^6}
\newunicodechar{⁷}{^7}
\newunicodechar{⁸}{^8}
\newunicodechar{⁹}{^9}
\newunicodechar{ⁿ}{^n}
\newunicodechar{∈}{\in}
\newunicodechar{∉}{\notin}
\newunicodechar{⊂}{\subset}
\newunicodechar{⊃}{\supset}
\newunicodechar{⊆}{\subseteq}
\newunicodechar{⊇}{\supseteq}
\newunicodechar{⊄}{\nsubset}
\newunicodechar{⊅}{\nsupset}
\newunicodechar{⊈}{\nsubseteq}
\newunicodechar{⊉}{\nsupseteq}
\newunicodechar{∪}{\cup}
\newunicodechar{∩}{\cap}
\newunicodechar{∀}{\forall}
\newunicodechar{∃}{\exists}
\newunicodechar{∄}{\nexists}
\newunicodechar{∨}{\vee}
\newunicodechar{∧}{\wedge}
\newunicodechar{⊼}{\bar{\wedge}}
\newunicodechar{⊽}{\bar{\vee}}
\newunicodechar{⌊}{\lfloor}
\newunicodechar{⌋}{\rfloor}
\newunicodechar{⌈}{\lceil}
\newunicodechar{⌉}{\rceil}
\newunicodechar{·}{\cdot}
\newunicodechar{∘}{\circ}
\newunicodechar{×}{\times}
\newunicodechar{↑}{\uparrow}
\newunicodechar{↓}{\downarrow}
\newunicodechar{→}{\rightarrow}
\newunicodechar{←}{\leftarrow}
\newunicodechar{⇒}{\Rightarrow}
\newunicodechar{⇐}{\Leftarrow}
\newunicodechar{↔}{\leftrightarrow}
\newunicodechar{⇔}{\Leftrightarrow}
\newunicodechar{↦}{\mapsto}
\newunicodechar{∅}{\emptyset}
\newunicodechar{∞}{\infty}
\newunicodechar{≅}{\cong}
\newunicodechar{≈}{\approx}
\newunicodechar{ℓ}{\ell}
\newunicodechar{↪}{\hookrightarrow}
\newunicodechar{⟨}{\langle}
\newunicodechar{⟩}{\rangle}

\newunicodechar{ℂ}{\mathbb{C}}
\newunicodechar{ℍ}{\mathbb{H}}
\newunicodechar{ℕ}{\mathbb{N}}
\newunicodechar{ℙ}{\mathbb{P}}
\newunicodechar{ℚ}{\mathbb{Q}}
\newunicodechar{ℝ}{\mathbb{R}}
\newunicodechar{ℤ}{\mathbb{Z}}

\newunicodechar{α}{\alpha}
\newunicodechar{β}{\beta}
\newunicodechar{γ}{\gamma}
\newunicodechar{Γ}{\Gamma}
\newunicodechar{δ}{\delta}
\newunicodechar{Δ}{\Delta}
\newunicodechar{ε}{\varepsilon}
\newunicodechar{ζ}{\zeta}
\newunicodechar{η}{\eta}
\newunicodechar{θ}{\theta}
\newunicodechar{Θ}{\Theta}
\newunicodechar{ι}{\iota}
\newunicodechar{κ}{\kappa}
\newunicodechar{λ}{\lambda}
\newunicodechar{Λ}{\Lambda}
\newunicodechar{μ}{\mu}
\newunicodechar{ν}{\nu}
\newunicodechar{ξ}{\xi}
\newunicodechar{Ξ}{\Xi}
\newunicodechar{π}{\pi}
\newunicodechar{Π}{\Pi}
\newunicodechar{ρ}{\rho}
\newunicodechar{σ}{\sigma}
\newunicodechar{Σ}{\Sigma}
\newunicodechar{τ}{\tau}
\newunicodechar{υ}{\upsilon}
\newunicodechar{ϒ}{\Upsilon}
\newunicodechar{φ}{\varphi}
\newunicodechar{ϕ}{\phi}
\newunicodechar{Φ}{\Phi}
\newunicodechar{χ}{\chi}
\newunicodechar{ψ}{\psi}
\newunicodechar{Ψ}{\Psi}
\newunicodechar{ω}{\omega}
\newunicodechar{Ω}{\Omega}

\newunicodechar{ℋ}{\mathcal{H}}

\usepackage[english]{babel}
\usepackage{xcolor}
\usepackage[vlined,linesnumbered]{algorithm2e} %
\usepackage{amsmath, amsthm, amssymb, mathtools, thmtools}
\usepackage[textsize=footnotesize]{todonotes}
\usepackage{hyperref}
\usepackage[capitalise]{cleveref}

\crefname{enumi}{property}{properties}

\newcommand{\calS}{\mathcal{S}}
\newcommand{\calT}{\mathcal{T}}
\newcommand{\calU}{\mathcal{U}}
\newcommand{\U}{\calU}
\newcommand{\calX}{\mathcal{X}}

\newcommand{\F}{\mathbb{F}}

\newcommand{\OO}{\mathcal{O}}
\newcommand{\tOO}{\tilde{\OO}}
\newcommand{\oo}{o}

\newcommand{\OPT}{\text{\upshape\textsc{opt}}}

\newcommand{\Prp}[1]{\Pr\left[#1\right]}
\newcommand{\Ep}[1]{\mathbb{E}\left[#1\right]}

\newcommand{\Var}[1]{\operatorname{Var}\left[#1\right]}

\def\D{\mathcal{D}}

\renewcommand{\epsilon}{\varepsilon}

\DeclareMathOperator{\poly}{poly}
\DeclareMathOperator{\polylog}{polylog}

\DeclareMathOperator{\row}{row}
\DeclareMathOperator{\column}{column}
\DeclareMathOperator{\rowg}{rGroup}
\DeclareMathOperator{\columng}{cGroup}
\DeclareMathOperator{\rowidx}{rowIndex}

\DeclarePairedDelimiter{\set}{\{}{\}}
\DeclarePairedDelimiter{\indicator}{[}{]}
\def\size#1{|#1|}
\DeclarePairedDelimiter{\paren}{(}{)}
\DeclarePairedDelimiter{\ceil}{\lceil}{\rceil}
\DeclarePairedDelimiter{\floor}{\lfloor}{\rfloor}

\newtheorem{theorem}{Theorem}
\newtheorem{lemma}[theorem]{Lemma}
\newtheorem{corollary}[theorem]{Corollary}
\newtheorem{observation}[theorem]{Observation}

\theoremstyle{definition}
\newtheorem{definition}[theorem]{Definition}

\SetFuncSty{textsc}
\SetKwProg{proc}{Procedure}{}{}
\DontPrintSemicolon

\def\keySum{\textsf{keySum}\xspace}
\def\Count{\textsf{count}\xspace}

\usepackage{tikz}
\usetikzlibrary{calc,arrows.meta,shadows}

\newif\iffinal
\finaltrue

\title{Stuffed IBLTs: Optimal Linear Multiset Sketches} %

\author[1]{Jonas Klausen}
\author[2]{Rasmus Pagh}
\author[3]{Stefan Walzer}

\affil[1]{Max Planck Institute for Informatics}
\affil[2]{University of Copenhagen}
\affil[3]{Karlsruhe Institute of Technology}

\date{\today}

\begin{document}
\pagenumbering{gobble}
\maketitle

\begin{abstract}
    A \emph{linear sketch} is a randomized linear mapping of a vector $v$ to a lower dimensional sketch vector, designed to preserve relevant information about $v$.
    We consider sketches of vectors $v ∈ ℤ^u$ (for $u ∈ ℕ$), designed for exact recovery of $v$ from its sketch.
    Concretely, our \emph{Stuffed IBLT} is a linear sketch configured with a capacity $n ∈ ℕ$ and a multiplicity limit $L ∈ ℕ$ and will recover $v$ with high probability whenever $‖v‖₀ ≤ n$ and $‖v‖_∞ ≤ L$.
    The sketch can be maintained efficiently under unrestricted updates to $v$, i.e., $v$ is not subject to any constraints in between decoding requests. 
    This makes the sketch useful for streaming algorithms and for solving the (multi)set reconciliation problem.

    For any positive constants $c$, $\epsilon$, and for large enough $n$ and $u \geq n^{1+\Omega(1)}$, the space usage of a Stuffed IBLT is within a factor $1+\epsilon$ from the information-theoretic optimum while allowing updates in constant time, and decoding in time $\OO(n)$ with failure probability $n^{-c}$.
    This improves the space/time/error probability trade-off over all prior constructions with similar functionality, including the Invertible Bloom Lookup Table (IBLT).
    The performance of the Stuffed IBLT is essentially the best we could hope for, up to the dependence on $c$ and $\epsilon$.
    We make the dependence on these parameters explicit, and further show a lower bound demonstrating that the dependence on $c$ is optimal within the class of peeling-based approaches.
    Our improvement comes from a careful combination of Walzer's spatial coupling technique (SODA~'21), the purity heuristic of Houen, Pagh, and Walzer~(SOSA~'23), and backyarding (Belazzougui, Kucherov, and Walzer, ESA~'24; Fleischhacker, Green Larsen, Obremski, and Simkin, ICALP~'24), allowing us to eliminate bottlenecks of past approaches.
    Unlike most prior work we do \emph{not} assume access to fully random hash functions.
\end{abstract}

\newpage
\pagenumbering{arabic}

\section{Introduction}

Let $\U$ be a universe of elements, which are called \emph{keys}. Consider the space $ℤ^{\U}$ of all functions from $\U$ to $ℤ$, or equivalently, the set of all \emph{signed multisets} of keys from $\U$, where each $x ∈ \U$ has a possibly negative multiplicity.

A \emph{linear signed multiset sketch}, just \emph{sketch} from now on, stores a lossy representation $S(v)$ of $v ∈ ℤ^\U$ that supports updates to $v$ and a decoding operation. Decoding is meant to recover $v$, but may fail. The function $S$ is a linear function, typically of the form $S : ℤ^{\U} → G^m$ where $G$ is a finite group and $m$ is much smaller than $|\U|$. It may be chosen at random and involve hash functions.

For a given \emph{capacity} $n$ and a given \emph{multiplicity limit} $L$, this paper shows how to construct a sketch from which $v$ can be recovered with high probability whenever $‖v‖₀ ≤ n$ (i.e.\ $v$ contains at most $n$ non-zeroes) and $‖v‖_∞ ≤ L$.
Linearity of $S$ allows for the sketches $S(v₁)$ and $S(v₂)$ of $v₁,v₂ ∈ ℤ^\U$ to be added to obtain a sketch $S(v₁)+S(v₂) = S(v₁+v₂)$ of $v₁+v₂$.
To update a sketch $S(v)$ when $f ∈ ℤ$ (signed) copies of $x ∈ \U$ are added to $v$ we compute $S(v+f·e_x) = S(v) + f·S(e_x)$ where $e_x$ is the unit vector for $x$. Note that the update is fast if $S(e_x)$ is sparse. 
The restrictions $n$ and~$L$ on capacity and multiplicity do \emph{not} apply in between decoding requests.

For a fixed prime $p$ we consider computation on a word RAM with word size $w ≥ \log₂ p$ that can perform $\F_p$ field operations (arithmetic modulo $p$) in $\OO(1)$ time.
Our main result is a sketch for multisets with almost optimal time, space usage, and error probability:

\begin{restatable}{theorem}{maintheorem} \label{thm:main}
  Let $c ≥ 1$ be arbitrary. Then for any $n,L ∈ ℕ$ and prime $p$ with $\max\{2L+1,n\} ≤ p$ and $n ≥ n₀(c)$ large enough there exists a linear signed multiset sketch for keys from $\U = \F_p$ with the following properties:
  \begin{itemize}
      \item Decoding recovers $v ∈ ℤ^\U$ with probability $1-\OO(n^{-c})$ when $‖v‖₀ ≤ n$ and $‖v‖_∞ ≤ L$.
      \item Update time $\OO(c)$ and expected decoding time $\OO(cn)$.
      \item Space usage $(1+e^{-Ω(c)})n ⌈\log₂|\U|+\log₂(2L+1)⌉$ bits.
  \end{itemize}
  If \(n \leq \size{\U}^{1-\Omega(1)}\) then, for any fixed \(\epsilon < 1\), the space usage can be reduced to \((1+\OO(\epsilon))\OPT\) bits, where $\OPT$ refers to the information-theoretical minimum, increasing the update time to \(\OO(c/\epsilon^3)\) and the decoding time to \(\OO(cn/\epsilon^3)\).
\end{restatable}
Note that working with keys in $\F_p$ is mainly a convenience; any set of keys that we can efficiently map bijectively to/from a subset of $\F_p$ can be handled.
Theorem~\ref{thm:main} shows that it is possible to approach information-theoretical barriers for storing multisets in a linear sketch while having very efficient update and decoding operations.

\paragraph{Is \cref{thm:main} optimal?}
There are ways in which our construction is clearly best possible: We can get $\OO(1)$ update time and $\OO(n)$ expected decoding time paired with error probability $\OO(n^{-c})$ and space $(1+ε)\OPT$ for any $c = \OO(1)$ and any $ε = Ω(1)$. Looking more closely, we also have reason to believe that the the dependence on $c$ (when it is not $\OO(1)$) is optimal. (Note that we make no such claim regarding $ε$.)
Our evidence is that no \emph{peeling-based} approach can be better, where peeling is the main mechanism powering all known linear multiset sketches with decoding time $\tOO(n)$. We will postpone a precise definition of peeling-based sketches to \cref{sec:lower-bound}. The following is an informal version of what we have proven.
\begin{theorem}[Informal Version of \cref{thm:lower-bounds}] \label{thm:lower-bounds-informal}
    Every peeling-based multiset sketch has a parameter $c$ such that, under some reasonable assumptions
    \begin{itemize}
        • expected update times are $Ω(c)$ and expected decoding times $Ω(cn)$
        • failure probability is $n^{-\OO(c)}$
        • the sketch requires $(1+e^{-\OO(c)})n$ cells, where each cell stores an element of $\U$ and a multiplicity.
    \end{itemize}
\end{theorem}
We remark that some reasonable approaches are “mostly” peeling-based but side-step our definition in minor ways. In this sense \cref{thm:lower-bounds-informal} is not particularly robust. We still suspect that it points to a barrier that cannot be surpassed without fundamentally new ideas.

\paragraph{Applications.}
Sketches of sparse signed multisets have several applications, we mention a few. 
In \emph{straggler identification} \cite{eppstein2010straggler}, a gatekeeper can maintain a sketch of employee entries and exits and recover the small set of employees remaining at closing time; sketches from multiple gates can be combined by linearity. 
In \emph{set reconciliation}, two parties holding similar vectors $a,b\in\mathbb Z^{\U}$ can exchange sketches, compute $S(a-b)$, and recover the differences with communication proportional to the number of discrepancies rather than proportional to $|\U|$ \cite{minsky2003set,goodrich2011invertible,eppstein2011setreconciliation,yang2024practicalRateless}. 
Similar ideas enable efficient comparison of large but highly similar genomic datasets represented as multisets of $k$-mers \cite{ShibuyaBelazzouguiKucherov2022}. 
Multiset sketches also arise in cryptography, where they have been used as building blocks for advanced cryptographic primitives and encrypted compression schemes that exploit sparsity while preserving confidentiality~\cite{dodis2004fuzzy,fleischhacker2023invertible}.
Building on a cryptographic secure aggregation primitive, multiset sketches have been used in distributed and federated analytics with secure aggregation; for example, Google used sketches to collect information about new vocabulary in its GBoard keyboard while preserving privacy\footnote{\url{https://research.google/blog/improving-gboard-language-models-via-private-federated-analytics/}}.

\subsection{Previous Work}\label{sec:priorwork}

\begin{table}[htb]
    \centering
    \caption{
    Comparison of (multi)set sketches where $n$ is the capacity $n$ and $c$ a parameter. The failure probability applies to decoding and is zero for the deterministic approaches.
    Running times assume a word RAM with constant time operations on finite fields.
    For simplicity space is measured in “cells” that hold field or ring elements, each representing an element of~$\U$, as well as possibly a multiplicity and possibly a checksum.
    Notes indicate whether the sketch is meant for sets (\textsc{s}), key-value pairs (\textsc{kv}), or multisets~(\textsc{m}), whether fully random hash functions are assumed (\textsc{h}), and whether cells use additional space for checksums~(\textsc{c}).
    }
    
    \label{tab:comparison}
    \begin{tabular}{rllccc>{\scshape}l@{\ }>{\scshape}c@{\ }>{\scshape}c}
    \toprule
    name
    & reference
    & cells
    & \makecell{update\\cost}
    & decoding
    & \makecell{failure\\rate}
    & \multicolumn{3}{c}{notes}
    \\
    \midrule
        characteristic poly.
        & \cite{minsky2003set}
        & $n$
        & $n$
        & $n^3$
        & $0$
        & s
        \\
        pinsketch
        & \cite{dodis2004fuzzy}
        & $n$
        & $n$
        & $n^2$
        & $0$
        & s
        \\
        randomized $k$-set
        & \cite{ganguly2007randomised}
        & $\OO(cn \log n)$
        & $c\log n$
        & $n$
        & $n^{-\Omega(c)}$
        & m && c
        \\
        IBF
        & \cite{eppstein2010straggler}
        & $\OO(cn)$
        & $c$
        & $n$
        & $2^{-\Omega(c)}$
        & m & h & c
        \\
        IBLT
        & \cite{goodrich2011invertible}
        & $\OO(cn)$
        & $c$
        & $n$
        & $n^{-\Omega(c)}$
        & kv & h & c
        \\
        SRS
        & \cite{DBLP:journals/tcs/BarkayPS15}
        & $2n$
        & $\polylog(n)$
        & $n \polylog(n)$
        & 0
        & m
        \\
        simple set sketch
        & \cite{baek2023simple}
        & $\OO(n)$
        & $3$
        & $n$
        & $\OO(n^{-1})$
        & s & h
        \\
        IBLT with a stash
        & \cite{belazzougui2024better}
        & $\OO(n)$
        & $c+3$
        & $n$
        & $2^{-\Omega(c)}$
        & s & h
        \\
        stacked IBLT
        & \cite{fleischhacker2023invertible}
        & $\OO(n)$
        & $c + \log\log n$
        & $n\log n$
        & $2^{-\Omega(c)}$
        & kv && c
        \\
        stuffed IBLT
        & {\bf new}
        & $(1+e^{-\Omega(c)})n$
        & $c$
        & $cn$
        & $n^{-\Omega(c)}$
        & m
        \\
        \bottomrule
    \end{tabular}
\end{table}

Table~\ref{tab:comparison} gives a cronological overview of the history of set and multiset sketches.
The first space-efficient linear sketches that can be used to reconstruct a sparse vector were based on algebraic techniques~\cite{minsky2003set,dodis2004fuzzy}.
These sketches only support \emph{sets} in the sense that the sketch of a frequency vector \(v\not\in\{0,1\}^\U\) does not in general contain enough information to recover $v$.
These sketches are deterministic, very space efficient, but also slow: 
The time to update is linear in the capacity, and moreover the decoding time obtained in~\cite{minsky2003set,dodis2004fuzzy} is, respectively, $n^3$ and~$n^2$.

The first sketch with fast updates was proposed by Ganguly~\cite{ganguly2007randomised}, furthermore supporting \emph{multisets}, i.e., decoding of sparse vectors $v ∈ ℤ^{\U}$.
This construction used \emph{hashing} to define a sparse sketch matrix, introducing a failure rate (depending on choice of parameters).
Eppstein and Goodrich~\cite{eppstein2010straggler} improved time and space usage when the failure rate is a small constant --- by independent repetition one can extend this to any desired failure rate.
However, they rely on an idealized assumption of fully random hash functions.
Goodrich and Mitzenmacher~\cite{goodrich2011invertible} built on similar ideas to improve the failure rate, with a sketch named the Invertible Bloom Lookup Table (IBLT).
They also explicitly emphasize the ability of their sketch to handle \emph{key-value pairs}, i.e., each key has a value associated with it.
While the IBLT from \cite{goodrich2011invertible} is not a linear multiset sketch (inserting a key with several distinct values makes the key unrecoverable) a ``fault tolerant'' version they also introduce could be used as one.

Houen, Pagh, and Walzer~\cite{baek2023simple} described a variant of the IBLT that avoids space overhead due to ``checksums'', but rather relies on the implicit weak checksum given by the hash values of keys. 
Like IBFs and IBLTs it is assumed that hash functions are fully random.
Unlike those data structures, linearity is over the field $\mathbb{F}₂$ so multisets are not supported (the multiplicity limit is $L=1$).

More recently, Belazzougui, Kucherov, and Walzer~\cite{belazzougui2024better} showed how to decrease the error probability of IBLTs with only a small increase in space usage by utilizing a small ``stash'' that allows the most likely failures to be corrected.
Fleischhacker, Green Larsen, Obremski, and Simkin~\cite{fleischhacker2023invertible} showed how to avoid the need for fully random hash functions that were assumed in previous IBF/IBLT constructions, at only a small cost in time complexity. However, their construction has a relatively large constant factor space overhead.

\paragraph{The multiplicity limit.}
We note that none of our predecessors explicitly mention a multiplicity limit $L$. Sometimes this limit is implicitly $1$ \cite{goodrich2011invertible,fleischhacker2023invertible}. Sometimes it is implicitly $∞$, meaning certain integer data types are assumed to be large enough to not overflow \cite{ganguly2006deterministic,eppstein2010straggler}. Two other papers \cite{baek2023simple,belazzougui2024better} construct sketches of $v ∈ \F₂^{\U}$ rather than sketches of $v ∈ ℤ^{\U}$, so there are only two possible multiplicities to begin with (two copies of the same key cancel out).

\paragraph{Relation to sparse recovery.}
Another conceptually related area is the literature on \emph{sparse recovery} and more generally \emph{compressed sensing}~\cite{gilbert2010sparse}. Like in our setting, these methods deal with linear sketches of vectors and the objective is to recover a sparse vector.
There are, however, critical differences that make our multiset recovery problem different (and more tractable).
First, sparse recovery is usually formulated over real-valued vectors and linear measurements, whereas our sketches represent signed multisets over a discrete universe (typically a finite field).
Second, a central goal in sparse recovery is often to output a good sparse \emph{approximation} to a vector that need not itself be sparse, with guarantees that are robust to measurement noise.
In contrast, we consider exact reconstruction in a noise-free discrete setting; this narrower task is what allows us to approach the information-theoretic space bound.

\subsection{Technical Overview}
\label{sec:contribution}
Our construction carefully combines many techniques that have been described in the literature.
Here we give a brief overview of each technique in isolation.

\begin{description}
    •[Basic Setup.]
    The randomised approaches \cite{eppstein2010straggler,ganguly2007randomised,goodrich2011invertible,baek2023simple,belazzougui2024better,fleischhacker2023invertible}
    generally follow a similar pattern. They use an array of $m ≥ n$ cells, initialised with zeroes, and each key $x ∈ \U$ is associated with a small set $H_x$ of cells via hash functions. %
    Any updates relating to $x$ lead to a specific value being added to each cell in $H_x$. If several keys contribute to the same cell, then the stored sum will not be informative on its own, but if a cell is \emph{pure}, i.e.\ only one key $x$ has contributed (with multiplicity in $[-L,L]$), then $x$ and its multiplicity can be inferred from the cell's content. Decoding a sketch relies on detecting such pure cells.
    •[Cells store \keySum and \Count.]
    Each cell of the sketch stores two numbers, \Count and \keySum. Adding $f$ copies of $x ∈ \U$ to a cell increments \Count by $f$ and \keySum by $f·x$. Note that multiplicities aggregate as desired, i.e.\ adding $f₁$ copies of $x$ and then $f₂$ copies of $x$ is the same as adding $f₁+f₂$ copies of $x$. A pure cell can be decoded by computing $f = \Count$ and $x = \keySum/\Count$. This method is used in \cite{eppstein2010straggler}, except that inserting only one copy at a time is intended.
    •[Purity Heuristic without Checksums.]
    Assuming \keySum and \Count are elements of the same field and $\Count ≠ 0$ it is always possible to compute \emph{some} candidate key $x = \keySum/\Count$ such that \Count copies of $x$ correspond to the cell content we observe.
    It is therefore usually impossible to decide whether a cell is pure when only looking at the cell's content. Most previous constructions addressed this problem by including a \emph{checksum} in each cell, either a sum of hash values \cite{eppstein2010straggler,goodrich2011invertible,fleischhacker2023invertible} or the sum of squares of keys \cite{ganguly2007randomised}. This comes at the cost of space and (arguably) elegance.

    An alternative purity heuristic uses the cell's index $i$ and checks whether $i ∈ H_x$, i.e.\ whether the candidate key $x$ actually hashes to the cell in question \cite{baek2023simple,belazzougui2024better}.
    We manage to leverage this heuristic while largely sidestepping the headache of analysing in detail what happens when it fails, as is done in \cite{baek2023simple,belazzougui2024better}.
    In turn, this is what allows us to avoid the space overhead associated with checksums.

    •[Decoding by peeling.]
    During decoding, whenever we identify a key and its multiplicity from a pure cell, we remove it from the sketch in hopes of rendering additional cells pure, which can then also be decoded. This iterative decoding procedure is called \emph{peeling} and was anticipated in \cite{eppstein2010straggler}, but the required mathematical toolkit to demonstrate its effectiveness is first invoked in \cite{goodrich2011invertible}.
    The underlying theory is most commonly stated in hypergraph language, with vertices corresponding to cells and hyperedges corresponding to the sets $H_x$ of cells associated with each key $x$.
    
    It is for instance known that a $(c=3)$-uniform hypergraph with $n$ uniformly random hyperedges is peelable with high probability when it has at least $1.23n$ vertices. For larger values of $c$ the threshold is higher, and generally it is $\Omega(cn)$ (see~\cite{walzer21} for details).
    •[Spatial Coupling.]
    Surprisingly, the peeling threshold can be improved when~$H_x$ contains $c$ random positions within close distance of one another rather than $c$ positions chosen uniformly at random, as shown by Walzer~\cite{walzer21}.
    The underlying technique of spatial coupling originated from coding theory \cite{KRU2011coupling} and has already been successfully applied to construct more space-efficient Bloom filters \cite{GL2022fusefilter}.

    •[Backyarding.]
    Decode sometimes recovers \emph{most} keys contained in a sketch before getting stuck. Maintaining a second sketch can then help to recover what is left \cite{fleischhacker2023invertible,belazzougui2024better}.
    Such a \emph{backyard} sketch can use more space and stronger hash functions relative to its smaller capacity without compromising overall performance.

    •[Simulating Full Randomness using Split \& Share.] Some multiset sketches \cite{eppstein2010straggler,goodrich2011invertible,baek2023simple,belazzougui2024better}, though not all \cite{ganguly2007randomised,fleischhacker2023invertible}, rely on the simple uniform hashing assumption. We simulate full randomness using an explicit construction from \cite{DBLP:conf/stoc/ChristianiPT15} and use the split-and-share trick \cite{dietz2009splittingtrick} to reduce the effect on our memory consumption to a lower order term.
    •[Quotienting.]
    It is possible to save space using a trick dating back to Cleary \cite{cleary84compact}: Assume we partition the universe into parts $\U₁,…,\U_{2^b}$ according to the $b$ most-significant bits of each key's binary representation. If we then use one sketch per part, each designed for a universe of size $|\U|/2^b$, we save $b$ bits per cell.
    However, this only works if the input contains roughly the same number of keys from each part. This can be ensured by shuffling the input with a suitable invertible hash function, called a \emph{quotient hash function}. A complication in our case is that the bits that can be saved by the Cleary trick are also used by the purity heuristic.
\end{description}

\subsection{Outline}
\begin{itemize}
    • We begin with a high level overview of our construction, explaining in what order the techniques listed in \cref{sec:design} are combined, what difficulties arise, and how these are addressed.
    • In \cref{sec:preliminaries}, apart from introducing some terminology, we state lemmas regarding families of hash functions that will be repeatedly used throughout our construction.
    • \Cref{sec:datastructure} contains our main contribution: Starting with a basic sketch design based on \cite{walzer21} in \cref{sec:coupled-iblt}, and building several layers of data structures around it in \cref{sec:coupled-oracle,sec:split-coupled,sec:quotient-sketch,sec:backyard} we finally prove \cref{thm:main} in \cref{sec:maintheorem}.
    • In \cref{sec:lower-bound}, we formally define and discuss the class of \emph{idealised peeling-based multiset sketches}. Lower bounds regarding this class of approaches yields \cref{thm:lower-bounds}, the formal version of \cref{thm:lower-bounds-informal}.
\end{itemize}
In the appendix we unpack results from three existing papers to adapt them to our specific needs.
\begin{itemize}
    • In \cref{sec:importing-spatial-coupling} we strengthen \cite[Theorem 1]{walzer21} in two minor ways.
    • In \cref{sec:hashing-appendix,sec:quotienthashing} we adapt families of hash function from \cite{DBLP:conf/latin/DemaineHPP06,DBLP:conf/stoc/ChristianiPT15}.
\end{itemize}

\section{Overview of our Construction} \label{sec:design}

Our construction for \cref{thm:main} uses a complicated multi-level architecture combining many techniques from previous work (as listed in \cref{sec:contribution}).

Our starting point is an idealised peeling-based sketch that has three issues. We modify it in a sequence of four steps, each addressing one issue (the fourth step fixing an issue introduced by the previous steps).
This gives us a sequence of five sketches as listed in \cref{fig:overview} (not counting the auxiliary backyard structure), ending with the Stuffed IBLT. We now give a high-level overview.

\paragraph{Starting Point: Spatial Coupling.}
In \cite{walzer21}, we are given a distribution $\D$ on subsets of $[m]$ such that a hypergraph with $n$ hyperedges sampled independently from $\D$ is likely peelable, even when the number $m$ of vertices only slightly exceeds the number $n$ of hyperedges. Concretely, for hyperedge size $c$, we only need $m ≥ (1 + e^{-Θ(c)})n$. This immediately implies a sketch (called “Oracle” in \cref{fig:overview}) with a correspondingly low cell count that should, in principle, lend itself to being decoded by peeling.

Specifically, the Oracle-Sketch consists of a collection of \(m\) cells (corresponding to the vertices of the hypergraph), and each key is associated with a random hyperedge sampled from \(\D\) and stored in the corresponding cells.
The contents of the Oracle-Sketch can then be recovered iff the hypergraph is peelable.

The Oracle-Sketch has three problems: It relies on a purity oracle to identify the peelable cells, its failure probability is \(\OO(1/n)\), and can't be decreased, and, finally, the sketch needs access to a fully random hash function, which dominates the overall space usage.

\paragraph{Step 1: Detecting Purity.}
When inspecting a cell, we can almost always infer a candidate key $x = \keySum / \Count$ that would account for the cell's content on its own.
If the cell in question is actually one of the cells associated with $x$ -- an event with prior probability $\OO(1/n)$ -- we have strong reason to believe that the cell really is pure.
This is the fundamental principle behind the \emph{Heuristic-Sketch}, which improves upon the Oracle-Sketch by getting rid of the purity oracle.

A purity heuristic relying on this check can suffer from false positives. The authors of \cite{baek2023simple,belazzougui2024better} link these false positives to bad structures, called \emph{anomalies} (see \cref{def:anomaly}), of which there are only constantly many in expectation.
We use similar arguments, although they are made slightly more technical by the fact that the field $\F_p$ we use for $\keySum$ and the group $ℤ_{2L+1}$ used for \(\Count\) have different sizes.
This leads to several cases that our analysis must handle (see \cref{sec:risk-of-anomalies}), but the asymptotic probability of an anomaly occurring remains unchanged.

In contrast to \cite{baek2023simple,belazzougui2024better} we do not try to recover from errors caused by anomalies.
Instead we upper bound the damage in the following sense:
The Heuristic-Sketch introduces a layer of partitioning, where each key is stored in one of several sub-structures \(T_1, \dots, T_{\sqrt{n}}\).
The \(T_i\)'s are decoded separately such that each anomaly only corrupts the output from one sub-structure \(T_i\).

The Heuristic-Sketch will typically fail to recover a small number of \emph{leftover} keys, which introduces the need for the Backyard described in step 4 below.
To show that the number of leftover keys is concentrated around its (small) expectation, we exploit that the failure events of the substructures (caused by anomalies or non-peelability) are \emph{negatively associated}.

Sketches that fail to completely recover the stored set are said to have non-zero \emph{error tolerance}, which we formally define in \cref{sec:preliminaries}.

\paragraph{Step 2: Reducing the Error Probability.}
One bottleneck of the Heuristic-Sketch is that the hash function used for distributing keys across the sub-substructures fails with probability \(\OO(1/n)\).
Boosting the error probability would lead to slower evaluation, which is incompatible with the fast purity heuristic.

We therefore introduce a second layer of partitioning via a new hash function that fails only with probability $\OO(1/n^c)$.
That is, the \emph{Split-Sketch} consists of several smaller Heuristic-Sketches, with each key stored in exactly one Heuristic-Sketch.
This second, slower hash function is crucially \emph{not} used for the purity tests.

Note that while the probability of each sub-sketch failing is too large, the probability that “too many” fail simultaneously can be bounded by $\OO(n^{-c})$ using Chernoff-like arguments.

\paragraph{Step 3: Quotienting and Sharing Hash Functions to Reduce Space.}
The Split-Sketch described in step 2 still relies on simulating full randomness on sets of $n$ keys, which takes up more space than the sketch itself.
But now that the error probability of each sketch is only $\OO(n^{-c})$, we can rely on union bounds over many structures to show that the \emph{same} hash function will likely work for all of them simultaneously. We hence use yet another layer of partitioning into smaller Split-Sketches that share a single set of hash functions.

Note that the hash function controlling the partitioning need not be fully random but only strong enough to ensure load balancing.
By using the \emph{quotient} function of \cite{DBLP:conf/latin/DemaineHPP06} for this purpose, we can further save (up to) \(\log_2 n\) bits per cell, by letting each cell store keys from the smaller universe \([\calU/n]\).

The partitioning given by the quotient function leads to some sub-structures being overloaded, and thus the keys stored in these sketches are unrecoverable.
Even then, we can give a sublinear bound on the number of leftover keys, at \(\OO(n^{1-\epsilon})\) for some small \(\epsilon\).

If not for the leftover keys, the resulting \emph{Quotient-Sketch} would have achieved all of our goals.

\paragraph{Step 4: Backyarding to Recover the Leftover Keys.}
Since we can bound the number of leftover keys by $\OO(n^{1-\epsilon})$, we can use a backyard sketch that uses a factor of $\OO(c)$ more cells than its capacity and fully random hashing on its key set without compromising our overall memory budget. A simple variant of a sketch from \cite{goodrich2011invertible} is strong enough to serve as a backyard in the Stuffed IBLT.

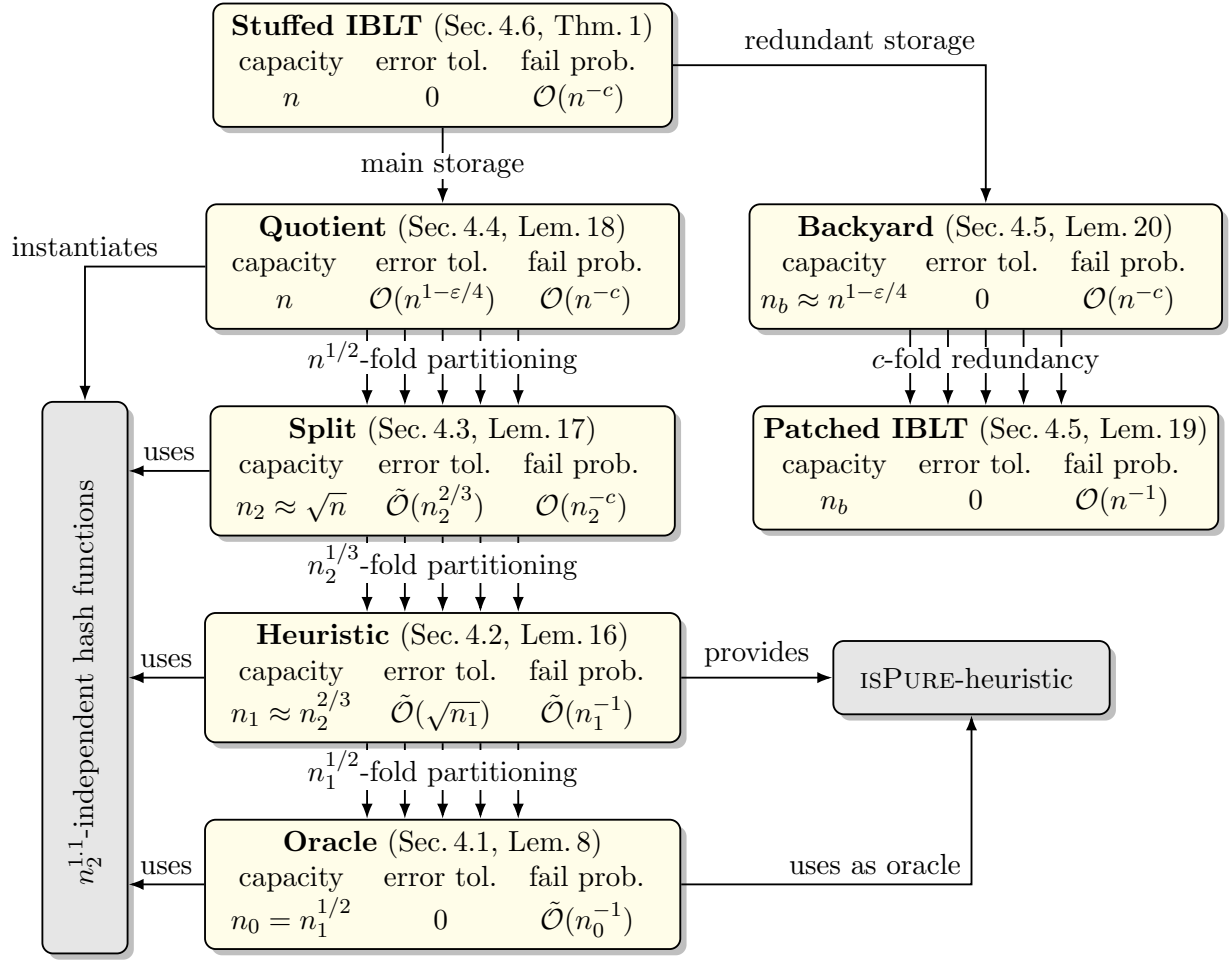
\begin{figure}[htbp]
  \def\perfTable#1#2#3{
    \begin{tabular}{ccc}
        capacity & error tol. & fail prob.\\
        #1 & #2 & #3
    \end{tabular}
}
\tikzstyle{sketch}=[draw,rounded corners,fill=yellow!10,align=center,drop shadow]
\tikzstyle{auxiliary}=[draw,rounded corners,fill=gray!20,align=center,drop shadow]
\def\sketchName#1#2#3{
    \textbf{#1} (Sec.\,\ref{#2}, Lem.\,\ref{#3})
}
\def\sketchNameThm#1#2#3{
    \textbf{#1} (Sec.\,\ref{#2}, Thm.\,\ref{#3})
}
\tikzstyle{overlay}=[fill=white,inner sep=0]
\def\arrowBatch#1#2#3{
    \foreach \i in {-2,-1,0,1,2}{
        \path (#1.south) ++ (\i*0.5,0) coordinate (from);
        \path (#2.north) ++ (\i*0.5,0) coordinate (to);
        \draw[->] (from) -- (to);
    }
    \path ($(#1.south)!0.5!(#2.north)$) node[overlay,anchor=base] {#3};
}
\begin{tikzpicture}[>=stealth,>={Latex[length=2mm,width=1.5mm]},line width=0.7pt]
    \node[sketch] (result) {
        \sketchNameThm{Stuffed IBLT}{sec:maintheorem}{thm:main}\\
        \perfTable{$n$}{0}{$\OO(n^{-c})$}
    };
    \node[sketch,below=of result] (quotient) {
        \sketchName{Quotient}{sec:quotient-sketch}{quotientsketch}\\
        \perfTable{$n$}{$\OO(n^{1-ε/4})$}{$\OO(n^{-c})$}
    };
    \node[sketch,below=of quotient] (split) {
        \sketchName{Split}{sec:split-coupled}{splitsketch}\\
        \perfTable{$n₂ ≈ \sqrt{n}$}{$\tOO(n₂^{2/3})$}{$\OO(n₂^{-c})$}
    };
    \node[sketch,below=of split] (heuristic) {
        \sketchName{Heuristic}{sec:coupled-oracle}{heuristicsketch}\\
        \perfTable{$n₁ ≈ n₂^{2/3}$}{$\tOO(\sqrt{n₁})$}{$\tOO(n₁^{-1})$}
    };
    \node[sketch,below=of heuristic] (oracle) {
        \sketchName{Oracle}{sec:coupled-iblt}{oraclesketch}\\
        \perfTable{$n₀ = n₁^{1/2}$}{$0$}{$\tOO(n₀^{-1})$}
    };
    \node[sketch,below right=of result] (backyard) {
        \sketchName{Backyard}{sec:backyard}{thm:backyard}\\
        \perfTable{\!\!\!$n_b ≈ n^{1-ε/4}$\!\!\!}{$0$}{$\OO(n^{-c})$}
    };
    \node[sketch,below=of backyard] (iblt) {
        \sketchName{Patched IBLT}{sec:backyard}{backyard-inner}\\
        \perfTable{$n_b$}{$0$}{$\OO(n^{-1})$}
    };
    \node[auxiliary,left=of heuristic,anchor=south,rotate=90,minimum width=7.3cm, inner sep=10] (hashes) {
        $n₂^{1.1}$-independent hash functions
    };
    \node[auxiliary,right=2cm of heuristic,inner sep=10] (ispure) {
        \textsc{isPure}-heuristic
    };
    \arrowBatch{quotient}{split}{$n^{1/2}$-fold partitioning}
    \arrowBatch{split}{heuristic}{$n₂^{1/3}$-fold partitioning}
    \arrowBatch{heuristic}{oracle}{$n₁^{1/2}$-fold partitioning}
    \arrowBatch{backyard}{iblt}{$c$-fold redundancy}
    \draw[->] (result) -| node[above left] {redundant storage} (backyard);
    \draw[->] (result) -- node[overlay] {main storage} (quotient);
    \draw[->] (quotient) -| node[above] {instantiates} (hashes);
    \foreach \s in {split,heuristic,oracle}{
        \draw[->] (\s) -- node[above] {uses} (hashes.south |- \s);
    }
    \draw[->] (heuristic) -- node[above] {provides} (ispure);
    \draw[->] (oracle) -| node[above left] {uses as oracle} (ispure);
\end{tikzpicture}
  \caption{Overview of our structures. With the exception of the backyard, sketches of capacity $N$ use \((1+e^{-\Omega(c)})N\) cells, support updates in \(\OO(c)\) time, and decoding in \(\OO(nc)\) time.}
  \label{fig:overview}
\end{figure}

\section{Preliminaries} \label{sec:preliminaries}
Each sketch we present is parameterized by two positive integer parameters \(n\) and \(c\), as well as a non-negative real \(\epsilon\):
The capacity \(n\) indicates the number of distinct keys that may be stored in the sketch at the time of decoding.
The quality parameter \(c\) roughly translates to the number of locations each key is mapped to in the sketch.
The use of more hash functions both allows us to make the sketch more compact and, for the sketches with tunable error probability, to decrease the probability of incorrectly retrieving the stored keys down to \(\OO(n^{-c})\).
At the same time, \(c\) affects the running time with updates taking \(\OO(c)\) time and decoding expected \(\OO(cn)\) time.
Some of our arguments only apply for sufficiently large \(n\), notably \cref{thm:threshold}, with the threshold being a function of \(c\).
Parameter \(c\) must thus be set independently of \(n\).

The \emph{stretch} parameter \(\epsilon\) controls the space overhead used by our sketches, and must likewise be set independently of \(n\).
A smaller \(\epsilon\) ensures lower space consumption at the cost of slower updates.

Throughout this section we use \(\calS \coloneqq \set{x \in \calU \colon f_x \neq 0}\) to denote the set of up to \(n\) keys stored in the sketch at the time of decoding, and the notation \(\tOO\) to hide factors polylogarithmic in \(n\) as well as factors depending only on \(c\).
In other words, \(\tOO(g(n)) \coloneq \OO(g(n)) \cdot (\log n)^{\OO(1)} \cdot f(c)\).
We will not, however, employ this notation when discussing running time, emphasizing that the time spent performing update and decode-operations only scale linearly with \(c\).

\paragraph{Error Tolerance and Failure Probability.}
In \cref{sec:coupled-oracle} onwards we will employ numerous (smaller) sketches as subroutines, and they may fail to decode correctly.
If a sketch of capacity \(n\) retrieves \(n\) keys and is still not empty, this is a clear sign that decoding failed.
In general, however, a sketch will not be able to verify its own output and will return up to \(n\) keys when failing -- an unknown number of which will belong to \(\calS\).

For sketches building on substructures, it will often not be possible to recover \(\calS\) exactly.
Let \(X\) be the set of key which are reported ``wrongfully'': that is, keys whose multiplicity in \(\calS\) differ from the multiplicity reported when decoding the sketch (either of which may be zero).
A sketch is said to have \emph{error tolerance} \(t\) and failure propability \(p\) if \(\Prp{\size{X} > t} \leq p\).
We say that the sketch is decoded correctly as long as \(\size{X} \leq t\).

Crucially, our final sketch has an error tolerance of zero, i.e.\ when it does not fail then it recovers \(\calS\) exactly.

\subsection{Hash Functions} \label{sec:hashfunctions}
In this short section we introduce the hash functions that we employ throughout the paper, and discuss their properties.
For details on their construction, we refer to \cref{sec:hashing-appendix,sec:quotienthashing}.

\begin{restatable}{lemma}{hashfunctionstwo} \label{hashfunctions2}
  For any \(\epsilon < 1\) and positive integers \(m\) and \(k\) such that \(k \leq m \leq \calU\), and \(r \leq m^{\OO(1)}\) there exists a family of hash functions \(h : \calU \to [r]\) with the following properties:
  \begin{enumerate}
  \item The space usage of \(h\) is \(\OO(km^{\epsilon}  \epsilon^{-3} \log m + \log\size{\calU})\) bits.
  \item Constructing and evaluating \(h\) takes time \(\OO(\epsilon^{-3})\).
  \item Any fixed set of \(\ell \leq m\) keys is hashed \(k\)-independently with probability \(1 - \OO(\ell^2/m^{3})\). \label{hashfunctions2:independence}

  \item Any iteratively defined set (see \cref{iterativelyDefinedDef} below) of \(\ell \leq k\) keys is hashed independently with probability \(1 - \OO(\ell^2/m^3)\). \label{hashfunctions2:iterativelydefined}
  \end{enumerate}
  We refer to \(m\) as the \emph{capacity} of the function and to \(k\) as its \emph{independence}.
\end{restatable}

Intuitively, the capacity \(m\) should be thought of as the primary parameter of the family, and will coincide with the capacity of the sketches described in the following sections.
Independence \(k\) mirrors the classic notion of \(k\)-independence, except that we incur a small additive error for each set that we claim independence on.
The space parameter \(\epsilon\) will be set to a constant value in the majority of our applications, with the exception of The Backyard (\cref{sec:backyard}).

Iteratively defined sets are the sets of keys that can be defined as a function of a limited number of observed hash values, formalized as follows:
\begin{restatable}[Iteratively Defined Sets]{definition}{iterativelydefinedDef} \label{iterativelyDefinedDef}
  A random set of keys \(\set{x_1, \dots, x_\ell} \subset \calU\) is said to be iteratively defined with respect to hash function \(h\) if, for each \(i \in [\ell]\),
the key \(x_i\) is determined as a function of prior keys \(x_1, \dots, x_{i-1}\) and their associated hash values \(h(x_1), \dots, h(x_{i-1})\).  
\end{restatable}

For the sketches with tunable failure probability, we will need \cref{hashfunctions2}~(\ref{hashfunctions2:independence}) to hold with higher probability \(1 - \OO(n^{-c})\).
This is achieved by instantiating \(\OO(c)\) hash functions from \cref{hashfunctions2}, and using their sum as the hash value.
The probability that all \(\OO(c)\) functions fail to be independent on the set of keys in question is then bounded by \(\OO(n^{-c})\).

Throughout the construction of our sketch, we need to show that the number of keys assigned to specific sub-structures is concentrated around its mean, even when the assignments are done by a hash function of limited independence.
We employ the following result due to \cite{DBLP:journals/siamdm/SchmidtSS95}.
\begin{lemma}[{\cite[Theorem 5.I.b]{DBLP:journals/siamdm/SchmidtSS95}}] \label{concentration}
  Let \(X\) be the sum of \(k\)-wise independent random variables taking values in \([0, 1]\), and define \(\mu \coloneq \Ep{X}\).
  For any \(0 \leq \delta \leq 1\) s.t. \(\delta^2\mu \leq k\),
  \[
  \Prp{X \geq (1+\delta)\mu} \leq \exp(- \floor{\delta^2\mu/3}) \, .
  \]
\end{lemma}

Finally, we need the following \emph{quotient} hash function, which allows us to recover \(x\) from \(h(x)\). For a proof, see \cref{sec:quotienthashing}.
\begin{restatable}[variant of {\cite[Theorem 5]{DBLP:conf/latin/DemaineHPP06}}]{theorem}{quotientfunctions} \label{quotientfunctions}
  For positive integers \(m\), \(c'\) and \(b\) such that \(m^{31/32} < b < m\) there exists a family of invertible hash functions \(h \colon \calU \to [b] \times [\size{\calU}/b]\) with the following properties:
  \begin{enumerate}
  \item Evaluating \(h\) and \(h^{-1}\) both takes time \(\OO(c')\).
  \item \(h\) and \(h^{-1}\) use \(\OO(m^{3/4} \log \size{\calU})\) bits of space.
  \item For any set \(S \subseteq \calU\) of \(m\) keys, and any \(\delta < 1\), the following holds with probability \(1 - \OO(n^{-c'})\):

    For \(i \in [b]\), let \(S_i = \set{x \in S \colon h_1(x) = i}\) be the keys mapped to bucket \(i\), and let \(B = \set{i \in [b] \colon \size{S_i} > (1+\delta)m/b}\) be the buckets assigned more than \((1+\delta)m/b\) keys.
    Then
    \[\sum_{i \in B} \size{S_i} \leq m \cdot \exp(-\delta^2m/(3b)) + \OO(m^{31/32}) \, .\]
  \end{enumerate}
\end{restatable}

\section{The Data Structure} \label{sec:datastructure}

In this section we prove \cref{thm:main}. As outlined in \cref{sec:design} and \cref{fig:overview}, this involves constructing a sequence of five auxiliary sketches presented in \crefrange{sec:coupled-iblt}{sec:backyard}, culminating in the final result sketch in \cref{sec:maintheorem}.

\subsection{A Sketch with a Purity Oracle} \label{sec:coupled-iblt}

We first describe the Oracle-Sketch from \cref{fig:overview}, which is an \emph{idealised} peeling-based multiset sketch that builds on the technique of spatial coupling \cite{walzer21}. It assumes access to a purity oracle, which will be replaced with a purity heuristic in \cref{sec:coupled-oracle}.
We invoke the following strengthened version of \cite[Theorem 1.1]{walzer21}, where \(n\) corresponds to the capacity of the Oracle-Sketch, and \(c\) is its quality parameter.
We prove \cref{thm:threshold} in \cref{sec:importing-spatial-coupling}.
\begin{restatable}{theorem}{threshold} \label{thm:threshold}
    For every $c ∈ ℕ$ and $n ≥ n₀(c)$ large enough, there exists $m ∈ ℕ$ and a distribution $\D$ on subsets of $[m]$ of size at most $c$ such that
    \begin{enumerate}
        \item $m = (1+e^{-Ω(c)})n$,
        \item \label{it:failure-to-peel} a hypergraph with vertex set $[m]$ and $n$ hyperedges sampled uniformly from $\D$ is peelable with probability $1-\OO(1/n)$,
        \item \label{it:form-of-hash-set} $H \sim \D$ arises as $H = \{h₀ + h_i \mid i ∈ [c]\}$ from independent random variables $h₀,h₁,…,h_c$ where, for some $w ∈ [m]$, $h₀$ is uniformly distributed in $[m-w]$ and $h₁,…,h_c$ are uniformly distributed in $\{0,…,w-1\}$, and
        \item \label{it:balanced} for $H \sim \D$ and all $i ∈ [m]$ we have $\Pr[i ∈ H] ≤ \frac{c}{n}$.
    \end{enumerate}
\end{restatable}

Our sketch uses the same structure as the hypergraph in \cref{thm:threshold} with $m = (1+e^{-Ω(c)})n$ cells corresponding to the vertices of the hypergraph and stored keys corresponding to hyperedges.
Any $x ∈ \U$ is thus associated with a set of cells $H_x \sim \D$,
and contributes to all cells of \(H_x\).

To realise the distribution \(\D\) described in \cref{thm:threshold} (\ref{it:form-of-hash-set}) we need $c+1$ hash functions that behave like fully random functions on \(\calS\), the set of at most $n$ keys present during decoding.%
\footnote{Since our sketch is history independent, we need not worry about keys that were previously present but have since been removed.
Their hash values no longer affect the state of the sketch.}
By applying \cref{hashfunctions2} (with both capacity \(m\) and independence \(k\) equal to \(n\)) such hash functions can be evaluated in time $\OO(1)$ each, and take up $\OO(cn^{1.1} \log n + c\log\size{\calU})$ bits of space in total.
The probability that any of the functions fail to hash \(\calS\) independently is bounded by $\OO(cn^{-1})$ by \cref{hashfunctions2}~(\ref{hashfunctions2:independence}).
When the hash functions are independent on \(\calS\), the probability that the keys aren't peelable is bounded by \(\OO(1/n)\) by \cref{thm:threshold}~(\ref{it:failure-to-peel}).
Decoding the sketch using \textsc{decodeIPMS} from \cref{alg:ipms} will thus succeed with probability at least \(1 - \OO(cn^{-1})\).
This algorithm uses a purity oracle, which is invoked at most $m+cn ≤ \OO(cn)$ times: Once for each cell initially, and again once whenever a cell is updated.
If the sketch is still non-empty after \(\OO(cn)\) such invocations peeling has failed and we terminate the decoding procedure.

We have yet to explain what is stored in a cell. We use two numbers $\keySum ∈ \U = \F_p$ and $\Count ∈ \{-L,…,L\}$. Addition in $\Count$ should wrap-around much like in $ℤ_{2L+1}$, but using a different ground set. Deviating from standard semantics, we require that “$\bmod (2L+1)$” yields a result in $\{-L,…,L\}$.\footnote{This is relevant because we are later reinterpreting $\Count$-values as elements of $\F_p$, and $-1$ and $2L$ are different elements of $\F_p$ unless $2L+1=p$.}
We add $f∈ℤ$ copies of $x ∈ \U$ to a cell by adding $f$ to $\Count$ and $f·x$ to $\keySum$.
If a cell stores $x ∈ \U$ with multiplicity $f ∈ \{-L,…,L\} \setminus \{0\}$ and we know it to be pure, then we can compute $f = \Count$ and $x = \keySum/\Count$.
Note that $\Count$ is different from zero (as an element of $\F_p$) by the assumption that $p ≥ 2L+1$ which we make for \cref{thm:main}. To summarise:

\begin{lemma}[Oracle-Sketch] \label{oraclesketch}
  In the setting of \cref{thm:main}, the Oracle-Sketch is a sketch containing \((1+e^{-\Omega(c)})n\) cells. Updates take \(\OO(c)\) time. Decoding requires \(\OO(cn)\) time and \(\OO(cn)\) queries to an oracle that determines if a cell is pure. Decoding retrieves \(\calS\) with probability \(1 - \tOO(n^{-1})\).

  The hash functions of the Oracle-Sketch take up \(\tOO(n^{1.1} + \log\size{\calU})\) space.
\end{lemma}

\subsection{A Sketch with a Purity Heuristic} \label{sec:coupled-oracle}
The goal of the Heuristic-Sketch is to improve upon the Oracle-Sketch by swapping the purity oracle for the following purity \emph{heuristic}, which attempts to verify whether the cell \(i\) is pure.

\begin{algorithm}

  \SetKwFunction{isPure}{isPure}
  \proc{\isPure{$i$}}{
    $(\keySum,\Count) ← T[i]$ \tcp{access cell $i$}
    \lIf{$\Count = 0$}{\Return false}
    $x ← \keySum / \Count$\;
    \Return $i ∈ H_x$\;
  }
\end{algorithm}
The heuristic will always recognize a pure cell, but may mistakenly report that a cell containing several keys is pure.
The following definition captures the problem cases for the heuristic.
\begin{definition} \label{def:anomaly}
    An \emph{anomaly} of size $ℓ ≥ 3$ is given by a cell $i ∈ [m]$ and a set of keys $\{x₁,…,x_{ℓ-1}\}$ stored in the sketch with multiplicities $f₁,…,f_{ℓ-1}$ such that the following holds.
    \begin{enumerate}
        • The sum $f_ℓ \coloneqq f₁ + … + f_{ℓ-1} \bmod (2L+1)$ is not zero.
        • With $x_ℓ \coloneqq (f₁x₁+…f_{ℓ-1}x_{ℓ-1})/f_ℓ \in \F_p$ we have that every $x ∈ \{x₁,…,x_ℓ\}$ satisfies $i ∈ H_x$.
    \end{enumerate}
\end{definition}
Anomalies are exactly what fools this heuristic: Sets of several keys all placed in the same cell that make \((\keySum, \Count)\) look like \(f_\ell\) copies of key \(x_\ell\), with \(i \in H_{x_\ell}\). 
\begin{observation} \label{lem:heuristic}
    In the absense of anomalies, the purity heuristic can serve as an (unfailing) purity oracle.
\end{observation}

\paragraph{Coupled-Sketch.}
Let us use the name \emph{Coupled-Sketch} to refer to a variant of the Oracle-Sketch where we replace the oracle with the \textsc{isPure}-heuristic.

The Coupled-Sketch faces two issues: The probability that an anomaly occurs is large, and \textsc{isPure} takes \(\OO(c)\) time to evaluate which leads to a decoding time of \(\OO(c^2 n)\).
The \emph{Heuristic-Sketch} remedies these two issues by partitioning its cells into several substructures, isolated from one another.
Hence an anomaly appearing in one structure will not prevent the decoding of the other structures.
We allow a small number of structures to fail in this way, which ensures that all but a sublinear number of keys will be decoded correctly.
Further, the partitioning of the sketch allows us to implement \textsc{isPure} in expected constant time.

\paragraph{Heuristic-Sketch.}
For any quality parameter \(c \geq 1\), the Heuristic-Sketch of capacity \(n\)
consists of \(\sqrt{n}\) Coupled-Sketches \(T_1, \dots, T_{\sqrt{n}}\), each of capacity \(n_0 \coloneqq \sqrt{n} + \sqrt{3\ln n} \cdot n^{1/4} \leq (1+\oo(1))\cdot \sqrt{n}\) and quality \(c\).
The sketch also involves a hash function \(h_1 \colon \calU \to [\sqrt{n}]\) which partitions \(\calU\) across the Coupled-Sketches:
To add \(f\) copies of key \(x\) to the sketch, we add \(f\) copies of \(x\) to \(T_{h_1(x)}\).

Instead of supplying each Coupled-Sketch with its own set of hash functions for computing the positions \(H_x\) associated with key \(x\), we instantiate a single set of \(c+1\) functions (as described in \cref{sec:coupled-iblt}), denoted \(h_2\), which is shared across \(T_1, \dots, T_{\sqrt{n}}\).
Each function in \(h_2\) is instantiated through \cref{hashfunctions2} with capacity \(n\) and independence \(\OO(cn)\).
Thus \(h_2\) is still evaluated in \(\OO(c)\) time, and takes up \(c\cdot \OO(cn^{1.1} \log n + c\log \size{\calU}) \leq \tOO(n^{1.1} + \log\size{\calU})\) space.

We construct \(h_1\) from \cref{hashfunctions2} with the same parameters, \(m=n\) and \(k=\OO(cn)\).
Thus \(h_1\) can be evaluated in constant time, and updates to the structure still take \(\OO(c)\) time.

\subsubsection{Decoding the Heuristic-Sketch} \label{sec:decoding-heuristic-sketch}
Like the Oracle-Sketch, decoding of the Heuristic-Sketch is allowed to fail with probability \(\tOO(1/n)\).
Unlike the Oracle-Sketch, the Heuristic-Sketch is not required to completely recover \(\calS\) when decoding: a sublinear number of keys may be lost, and likewise a sublinear number of wrong keys may be returned in addition to the keys of \(\calS\).

All Coupled-Sketches are decoded in parallel, utilizing the purity heuristic \textsc{isPure} which we refine in the following section.
Each decode performs up to \(n_0\) updates to its own sketch and queries the heuristic \(\OO(cn_0)\) times.
The former takes \(\sqrt{n} \cdot cn_0 = \OO(cn)\) time across all sketches.
The latter we likewise bound to \(\OO(cn_0)\) per sketch in the following section.

Recall that we abort the decoding of \(T_i\) if it has not terminated after \(n_0\) rounds of peeling. This ensures that, even in failure, decoding of a Coupled-Sketch terminates in \(\OO(cn_0)\) time and, if it reports an incorrect result, the reported set contains at most \(n_0\) incorrect keys.

To correctly decode a Heuristic-Sketch we require two conditions to be met:
First, that all Coupled-Sketches contain at most \(n_0\) keys from \(\calS\).
And second, that all \(c+2\) hash functions of the sketch must be independent on key set \(A\), defined as the set of all keys whose hash values are evaluated during decoding of the Heuristic-Sketch.

First, note that \(A\) is a random variable that arises dynamically in the decoding process depending on the hash values of keys in \(\calS\) and the hash values of other keys already revealed to be in \(A\).
Thus \(A\) is an iteratively defined set of \(\OO(cn)\) keys (see \cref{iterativelyDefinedDef}), and by \cref{hashfunctions2}~(\ref{hashfunctions2:iterativelydefined}) the keys of \(A\) are all hashed independently of each other with probability at least \(1 - \tOO(n^{-1})\).

By \cref{concentration} the probability that any Coupled-Sketch is overfull is bounded by \(\OO(1/n)\) when the keys of \(\calS\) are hashed independently by \(h_1\).

In the remainder of this section we thus assume that both of these properties hold.
To simplify the presentation, we can assume that \emph{all} keys of the universe are hashed independently.
This only changes the distribution of keys outside of \(A\), which doesn't alter the behaviour of the sketch.

\paragraph{The Faster Heuristic.}
We implement the heuristic of \cref{lem:heuristic} in a lazy manner in order to speed up the (expected) evaluation time. It tests whether $T_i[j]$, i.e.\ the $j$th cell of the $i$th substructure is pure.
Importantly, it tests whether $h₁(x)$ maps $x$ to \(T_i\), before evaluating \(h_2\).

\begin{algorithm}
  \SetKwFunction{isPure}{isPure}
  \proc{\isPure{$i,j$}}{
    $(\keySum,\Count) ← T_i[j]$ \tcp{access cell $j$ of the $i$th substructure}
    \lIf{$\Count = 0$}{\Return false}
    $x ← \keySum / \Count$\;
    \lIf{$h₁(x) ≠ i$}{\Return false \tcp*[h]{wrong substructure? tested in $\OO(1)$ time}} 
    \Return $j ∈ h_2(x)$ \tcp{the expensive test taking $\OO(c)$ time}
  }
\end{algorithm}

We now bound the expected time spent in \textsc{isPure} while decoding a Coupled-Sketch \(T_i\) by $\OO(cn_0)$.
Since the heuristic is invoked $\OO(cn_0)$ times and each line except the last takes $\OO(1)$ time, all lines except the last contribute $\OO(cn_0)$ time.
Now consider the set of keys reaching the last line.
By saving \(h_2(x)\) for each such key during decoding, we only have to compute \(h_2(x)\) in time $\OO(c)$ once per key.\footnote{More precisely, whenever we compute $h₂(x)$ for a key $x$, we store $x$ in a hashtable $T₁$ to recognise such keys later on. Moreover, we store all pairs $(x,j)$ with $j ∈ h₂(x)$ in another hashtable $T₂$ so we can later efficiently decide for a given $j'$ whether $j' ∈ h₂(x)$, which is what line 6 requires. It is sufficient if $T₁$ and $T₂$ are hash tables with linear chaining using a universal hash function.}
These keys may include up to $n_0$ keys actually stored in the sketch, accounting for $\OO(cn_0)$ time again. It may also include some \emph{foreign} keys not stored in the sketch.

For each foreign key \(x\), $\Prp{h₁(x) = i} = 1/\sqrt{n}$ so only $\OO(cn_0/\sqrt{n}) = \OO(c)$ of these are expected to reach the last line.
Assuming our two success-conditions, the number of foreign keys reaching the last line is a sum of independent variables.
We can thus terminate the decoding after \(\OO(n_0)\) evaluations of the final line, and this will only increase the probability of failure by a negligible amount.
It will, however, ensure that decoding of the Coupled-Sketch is concluded in expected \(\OO(cn_0)\) time, even when our success-conditions fail to be met (including the case where the Heuristic-Sketch itself is filled beyond capacity).

Summing over all \(\sqrt{n}\) Coupled-Sketches, decoding is completed in expected time \(\OO(cn)\).

In \cref{sec:split-coupled,sec:quotient-sketch} we build upon the Heuristic-Sketch by adding more layers of partitioning.
The hash functions for such outer partitioning will \emph{not} be used for the purity heuristic.

\subsubsection{The Risk of Anomalies, and the Number of Leftover Keys} \label{sec:risk-of-anomalies}
We now bound the risk of an anomaly (\cref{def:anomaly}) occuring within each sketch \(T_i\).

\begin{lemma}
  \label{lem:probability-for-anomaly}
  When keys are hashed independently and $T_i$ is not overfull, then the probability that \(T_i\) contains an anomaly is bounded by \(\tOO(n^{-1/2})\).
\end{lemma}

\begin{proof}
Recall that an anomaly of size $ℓ$ is given by the index $j$ of a cell and a set \(\calX=\set{x_1, \dots, x_{\ell-1}} \subseteq \calS\) of keys stored in the sketch. Let \(f_1, \dots, f_{\ell-1}\) be the multiplicities of $x₁,…,x_{ℓ-1}$. From \cref{def:anomaly} we can derive another number $f_ℓ ∈ [-L,L]$ and a key $x_ℓ ∈ \U$ such that
\begin{align*}
    f₁+…+f_{ℓ-1} &= f_ℓ \pmod{2L+1}\\
    f₁x₁+…+f_{ℓ-1}x_{ℓ-1} &= f_ℓ x_ℓ \pmod{p}.
\end{align*}
We will consider several cases of anomalies. For each case, we either show that the occurence probability of such an anomaly in $T_i$ is bounded by \(\tOO(n^{-1/2})\), or we show that it implies the occurence of an anomaly with smaller $ℓ$.

\begin{description}
    •[Case 1: $x_ℓ ∉ \{x₁,…,x_{ℓ-1}\}$.]
In this case there are $ℓ$ distinct keys that each need to use cell $j$.
The probability for the anomaly to manifest in cell $j$ is therefore at most $(\frac{c}{n})^ℓ$ since each key has to hash to the right substructure (probability $1/\sqrt{n}$) and use cell $j$ within it (probability $≤ c/n₀ ≤ c/\sqrt{n}$ using \cref{thm:threshold} (\ref{it:balanced})). By a union bound over $ℓ ≥ 3$, all choices $\binom{n}{ℓ-1}$ for $\calX$ and all $(1+e^{-Ω(c)})n₀$ choices for $j$ within $T_i$ the probability for such an anomaly to arise in $T_i$ is bounded by
\begin{align*}
    \sum_{ℓ = 3}^n \binom{n}{\ell-1} (1+e^{-Ω(c)})n₀ \paren*{\frac{c}{n}}^{\ell}
    &\leq \OO(n_0) \cdot \sum_{\ell ≥ 3} \paren*{\frac{n e}{\ell-1}}^{\ell-1} \paren*{\frac{c}{n}}^{\ell} \\
    &\leq \OO(cn_0/n) \cdot \sum_{\ell ≥ 3} \paren*{\frac{ce}{\ell -1}}^{\ell-1}
    \leq \OO(n^{-1/2}) \cdot c^{\OO(c)} \, .
\end{align*}
•[Case 2: $x_ℓ ∈ \{x₁,…,x_{ℓ-1}\}$.] Without loss of generality assume $x_ℓ = x_{ℓ-1}$. We can rearrange the equations from earlier to get
\begin{align*}
    f₁+…+f_{ℓ-2} &= (f_ℓ-f_{ℓ-1}) \pmod{2L+1}\\
    f₁x₁+…+f_{ℓ-2}x_{ℓ-2} &= (f_ℓ-f_{ℓ-1}) x_{ℓ-1} \pmod{p}.
\end{align*}
Let also $f^* \coloneqq (f_ℓ-f_{ℓ-1}) \bmod (2L+1)$. We consider four sub-cases.
\begin{description}
    •[Case 2.1: $f_ℓ-f_{ℓ-1} = 0$.]
    Then we have $f₁+…+f_{ℓ-3} = -f_{ℓ-2} \pmod{2L+1}$ and $f₁x₁+…+f_{ℓ-3}x_{ℓ-3} = -f_{ℓ-2}x_{ℓ-2} \pmod{p}$ and we have an anomaly of size $ℓ-2$ given by $x₁,…,x_{ℓ-3}$.
    •[{Case 2.2: $f_ℓ-f_{ℓ-1} ∈ [-L,L]$}.]
    Then we have $f^* = f_ℓ-f_{ℓ-1}$ (i.e., “$\bmod\, (2L+1)$” has no effect) and we have an anomaly of size $ℓ-1$ given by $x₁,…,x_{ℓ-2}$ as \(f_1x_1 +\dots+ f_{\ell-2}x_{\ell-2} = f^* x_{\ell-1} \bmod{p}\).
    •[Case 2.3: $f_ℓ-f_{ℓ-1} < -L$.]
    We consider a slight variation to the definition of an anomaly, called an \emph{overflow anomaly}. It matches \cref{def:anomaly}, except that we replace the definition of \(x_\ell\) by
    \[x^+_\ell \coloneq f_1x_1 + \dots + f_{\ell-1}x_{\ell-1} / (f_\ell-(2L+1)) \in \F_p \, .\]
    and we further demand that 
\(x^+_\ell ∉ \set{x_1, \dots, x_{\ell-1}}\).
This final requirement ensures that we can bound the probability of an overflow anomaly occurring in exactly the same way as we bounded the occurrence of an anomaly with \(x_\ell \not\in \set{x_1, \dots, x_{\ell-1}}\) in Case 1.
In other words, the probability that an overflow anomaly occurs in \(T_i\) is at most \(\tOO(n^{-1/2})\).

We now return to the treatment of the case at hand where \(f_\ell - f_{\ell-1} < -L\).
As both \(f_\ell\) and \(f_{\ell -1}\) are in \([-L, L]\), it follows that $f_\ell - f_{\ell-1} ∈ [-2L,-L-1]$ and hence \(f_\ell - f_{\ell-1} = f^* - (2L+1)\).
We then have
\begin{align*}
    f₁+…+f_{ℓ-2} &= f^* \pmod{2L+1}\\
    f₁x₁+…+f_{ℓ-2}x_{ℓ-2} &= (f^*-(2L+1)) x_{ℓ-1} \pmod{p},
\end{align*}
certifying that \(\set{x_1, \dots, x_{\ell-2}}\) specifies an overflow anomaly of size $ℓ-1$.
Note that \(x_{\ell-1}\) is distinct from \(\set{x_1, \dots, x_{\ell-2}}\) by the definition of \(\calX\).
•[Case 2.4: $f_ℓ-f_{ℓ-1} > -L$.] This is handled analogous to Case 2.3 by introducing a corresponding notion of an underflow anomaly.\qedhere
\end{description}
\end{description}
\end{proof}

To summarise, if a Coupled-Sketch fails to decode it will be due to the presence of an anomaly or because decoding runs out of pure cells to peel (the failure event in \cref{thm:threshold}~(\ref{it:failure-to-peel})).
Hence the probability that Coupled-Sketch \(T_i\) fails to decode is bounded by \(\OO(1/n_0) + \tOO(n^{-1/2}) \leq \tOO(n^{-1/2})\).

\paragraph{The Number of Leftover Keys.}
In the rest of this section we bound the error tolerance of the Heuristic-Sketch by bounding the number of Coupled-Sketches that can fail simultaneously.
We say that a Coupled-Sketch is \emph{bad} if it contains an anomaly or decoding runs out of pure cells to peel. Let $B₁,…,B_{\sqrt{n}} ∈ \{0,1\}$ be the corresponding badness indicators.
Then \(\Ep{B_i} \leq \tOO(n^{-1/2})\) as discussed in the previous section.
  To obtain a good tail bound for the number \(B \coloneqq \sum_{i=1}^{\sqrt{n}} B_i\) of bad sketches we argue that $B₁,…,B_{\sqrt{n}}$ are negatively associated.
\begin{definition}[Negative association, \cite{JDP83negativeAssociation}] \label{def:NA}
  Random variables \(X_1, \dots, X_k\) are negatively associated if, for all disjoint \(A_1, A_2 \subset [k]\) and increasing functions \(f_1, f_2\),
  \[\mathrm{Cov}(f_1((X_i)_{i \in A_1}), f_2((X_j)_{j \in A_2})) ≤ 0 \]
  where $\mathrm{Cov}(Y,Z) := \Ep{(Y-\Ep{Y})(Z-\Ep{Z})}$.
\end{definition}
\begin{lemma} \label{lem:NA-applied}
  Indicators $B₁,…,B_{\sqrt{n}}$ are negatively associated when keys are hashed independently.
\end{lemma}
\begin{proof}
  Let $ℋ$ be the possible outcomes of $H_x$ for $x ∈ \U$ in the Heuristic-Sketch, i.e.\ the set of sets of cells that a key might be associated with.
  It is subdivided as $ℋ = ℋ₁ ∪ … ∪ ℋ_{\sqrt{n}}$ where $ℋ_i$ contains sets of cells in \(T_i\).
  For $x ∈ \calU$ and $H ∈ ℋ$ we define the indicator $I_{x,H} = \indicator{H_x = H}$. Now we can reason as follows.
  \begin{enumerate}
  \item For any $x ∈ \U$, the family $(I_{x,H})_{H ∈ ℋ}$ is negatively associated. To see this, we recall an argument found in \cite[Chapter 3]{Dubhashi_Panconesi_2009}. Consider any disjoint $ℋ₁,ℋ₂ ⊂ ℋ$, increasing functions $f₁,f₂$ and $Y = f₁((I_{x,H})_{H ∈ ℋ₁}), Z = f₂((I_{x,H})_{H ∈ ℋ₂})$. We may assume without loss of generality that $f₁(\vec{0}) = 0$ and $f₂(\vec{0}) = 0$ as shifting random variables by a constant does not affect covariance. Now note that in $(I_{x,H})_{H ∈ ℋ}$ all but one of the indicators are zero. Therefore $Y$ and $Z$ cannot be simultaneously non-zero. This implies
  \[ \Ep{YZ} = 0 ≤ \Ep{Z}\Ep{Y}, \text{ which is equivalent to } \mathrm{Cov}(Y,Z) ≤ 0. \]
  \item The family $(I_{x,H})_{H ∈ ℋ,x ∈ \U}$ is negatively associated, because it is the union of independent families of negatively associated random variables \cite[Property $P₇$]{JDP83negativeAssociation}. Here we use our assumption that all keys are hashed independently (see \cref{sec:decoding-heuristic-sketch}).
  \item When we fix the set $\calS$ of stored keys and their multiplicities, then $B_i$ is an increasing function of $(I_{x,H})_{H ∈ ℋ_i,x ∈ \U}$. That it is a function follows because the indicators fix everything related to $T_i$, including which keys from $\U$ hash to $T_i$ and which cells they use. The function is increasing because changing indicators from $0$ to $1$ means that additional keys are assigned to $T_i$, and both the failure to peel and the occurrence of an anomaly are clearly monotone properties.
  \item $B₁,…,B_{\sqrt{n}}$ arise from the family $(I_{x,H})_{H ∈ ℋ,x ∈ \U}$ of negatively associated random variables by partitioning this family into $\sqrt{n}$ disjoint subfamilies $(I_{x,H})_{H ∈ ℋ₁,x ∈ \U}, …, (I_{x,H})_{H ∈ ℋ_{\sqrt{n}},x ∈ \U}$ and applying an increasing function to each subfamily. This implies that $B₁,…,B_{\sqrt{n}}$ are negatively associated as desired \cite[Property $P₂$]{JDP83negativeAssociation}. \qedhere
  \end{enumerate}
\end{proof}

\begin{corollary}
  \label{cor:heuristic-failing-substructures}
  Let \(B\) be the number of bad Coupled-Sketches. Then
  $\Prp{B > \Theta(\log n)} ≤ \OO(1/n)$.
\end{corollary}
\begin{proof}
  By negative association \cite[Property \(P_2\)]{JDP83negativeAssociation} we have, for any \(t > 0\), that
  \[\Ep{e^{t B}} = \Ep{\prod_i e^{t B_i}} \leq \prod_i \Ep{e^{t B_i}}\]
  and thus the regular proof of the Chernoff bound goes through, showing that
  \[
  \Prp{B \geq (1+\delta) \Ep{B}} \leq 2^{-(1+\delta)\Ep{B}}
  \]
  for any \(\delta \geq 2e-1\).
  By linearity of expectation we have \(\Ep{B} = \sum_i \Ep{B_i} \leq \tOO(1)\), and by setting \(\delta = \Theta(\log n/\Ep{B})\) we obtain
  \(\Prp{B \geq \Theta(\log n)} \leq \OO(1/n)\) which proves the corollary.
\end{proof}

\begin{corollary}\label{heuristic-left-behind-keys}
  The Heuristic-Sketch has error tolerance \(\tOO(\sqrt{n})\) and failure probability \(\tOO(1/n)\).
\end{corollary}
\begin{proof}
  By \cref{cor:heuristic-failing-substructures} we may assume that only $\OO(\log n)$ of the Coupled-Sketches fail.
  A failing sketch may not recover its $n₀ = \OO(\sqrt{n})$ keys,
  and may in fact add $n₀$ erroneous keys.
  The product $\OO(\log n)·2n₀$ is $\tOO(\sqrt{n})$ as claimed.

  The arguments above assumed that all Coupled-Sketches contain at most \(n_0\) keys, and that all keys in \(A\) were hashed independently.
  As discussed in \cref{sec:decoding-heuristic-sketch}, these assumptions hold with probability \(1 - \tOO(1/n)\).
\end{proof}

Bringing together the observations at the top of \cref{sec:coupled-oracle} with the discussion in \cref{sec:decoding-heuristic-sketch} and \cref{sec:risk-of-anomalies}, we summarize the properties of the Heuristic-Sketch as follows:
\begin{lemma}[Heuristic-Sketch] \label{heuristicsketch}
  In the setting of \cref{thm:main}, the Heuristic-Sketch is a sketch on \((1+e^{-\Omega(c)})n\) cells with error tolerance \(\tOO(\sqrt{n})\) and failure probability \(\tOO(n^{-1})\).

  Updates take \(\OO(c)\) time while decoding takes \(\OO(cn)\) expected time.
  The hash functions of the Heuristic-Sketch take up \(\tOO(n^{1.1} + \log\size{\calU})\) space.
\end{lemma}

\subsection{A Sketch using Partitioning to Reduce Failure Probability}\label{sec:split-coupled}
We now describe the \emph{Split-Sketch} from \cref{fig:overview}.
By adding an additional layer of partitioning on top of the Heuristic-Sketch it achieves a tunable failure probability of \(\tOO(n^{-c})\).

For any integer \(c \geq 1\), the Split-Sketch of capacity \(n\) consists of \(n^{1/3}\) Heuristic-Sketches \(T_1, \dots, T_{n^{1/3}}\), each of capacity \(n_1 = (1 + \tOO(n^{-1/3})) n^{2/3}\) and quality \(c\) (we give a precise definition of \(n_1\) below).
The sketch involves a hash function \(h_3 \colon \calU \to [n^{1/3}]\) which distributes keys across the Heuristic-Sketches.
That is, to insert \(f\) copies of key \(x\) to the Split-Sketch, we insert the \(f\) copies into \(T_{h_3(x)}\).

Function \(h_3\) is the sum of \(c\) hash functions, each constructed from \cref{hashfunctions2} with capacity and independence both set to \(n\).
The functions making up \(h_3\) thus take up \(\tOO(n^{1.1} + \log\size{\calU})\) bits.
This is dominated by the \(n^{1/3}\) sets of hash functions from \(T_1, \dots, T_{n^{1/3}}\), which amounts to \(\tOO(n^{1.1} + n^{1/3} \log \size{\calU})\) bits in total.
Function \(h_3\) takes \(\OO(c)\) time to evaluate, keeping the update time of the sketch at \(\OO(c)\).
The Split-Sketch is decoded by decoding all Heuristic-Sketches in parallel.
As each sketch is decoded in expected \(\OO(cn_1)\) time, the Split-Sketch is decoded in expected \(\OO(n^{1/3} \cdot cn_1) = \OO(cn)\) time.

For the Split-Sketch to decode correctly, we require that none of the Heuristic-Sketches contain more than \(n_1\) keys.
By setting \(n_1 \coloneq (1+\delta)n^{2/3}\) for \(\delta = \sqrt{3 \ln(n^{c+1/3})}/n^{1/3} \in \tOO(n^{-1/3}\), it follows from \cref{concentration} that all Heuristic-Sketches are within capacity with probability at least \(1 - n^{-c}\), assuming that \(h_3\) is fully independent on \(\calS\) -- which is the case with probability at least \(1-\OO(n^{-c})\).

As each sketch is supplied with its own hash functions, they fail independently of each other and with probability \(\tOO(1/n_1)\).
The expected number of failing sketches is $\tOO(n^{1/3}/n₁) = \tOO(n^{-1/3})$,
and the probability that more than \(\OO(c \ln(n))\) fail at once is bounded by \(n^{-c}\) through a standard Chernoff bound (see \cref{cor:heuristic-failing-substructures}).

Each failing Heuristic-Sketch leaves at most \(2n_1\) keys behind (the $n₁$ keys present initially plus $n₁$ keys accidentally added to the sketch by misguided updates), while those that decode correctly have an error tolerance of \(\tOO(\sqrt{n_1})\) keys each.
In total, this adds up to an error tolerance of \(\tOO(c \ln(n) \cdot n_1 + n^{1/3} \cdot \sqrt{n₁}) = \tOO(n^{2/3})\) keys.

\begin{lemma}[Split-Sketch] \label{splitsketch}
  In the setting of \cref{thm:main} the Split-Sketch is a sketch on \((1+e^{-\Omega(c)})n\) cells with error tolerance \(\tOO(n^{2/3})\) and failure probability \(\OO(n^{-c})\).

  Updates take \(\OO(c)\) time while decoding takes \(\OO(cn)\) expected time.
  The hash functions of the Split-Sketch take up \(\tOO(n^{1.1} + n^{1/3} \log \size{\calU})\) space.
\end{lemma}

\subsection{A Sketch Saving Space through Quotienting and Sharing Hash Functions} \label{sec:quotient-sketch}

The Quotient-Sketch reduces the number of bits in each cell by \((1-\epsilon) \log n\) bits compared to the Split-Sketch, at the cost of increasing the running time of operations of the sketch by a factor of \(1/\epsilon\).
Without loss of generality we can assume that $\epsilon < 1/32$.
Independently, the Quotient-Sketch instantiates hash functions that are shared by all its sub-sketches, reducing the overall space for hashing to a sublinear number of bits.

The quotient function \(h \colon \calU \to [n^{1-\epsilon}] \times [\size{\calU}/n^{1-\epsilon}]\) (\cref{quotientfunctions}, with capacity \(m=n\), quality \(c' =c\) and \(b=n^{1-\epsilon}\) buckets) partitions keys across \(n^{1-\epsilon}\) Split-Sketches \(T_1, \dots, T_{n^{1-\epsilon}}\).
Each Split-Sketch is initialized with capacity \(n_2 = (1+\oo(1))n^\epsilon\) (defined precisely below) and quality parameter \((c+1)/\epsilon\).
Hence each Split-Sketch decodes correctly with probability \(1 - \OO(n^{-c-1})\) when not overfull.

Due to the nature of the quotient function \(h\), the sketches can treat keys as coming from a smaller universe of size \(\size{\calU}/n^{1-\epsilon}\), which reduces the space usage of each cell by \((1-\epsilon) \log n\) bits.
More precisely, a key \(x\) with \(h(x) = (i, x')\) is stored as \(x'\) in the \(i\)'th Split-Sketch \(T_i\).
When decoding the sketches, each key \(x'\) returned by the \(i\)'th sketch will be reported as \(h^{-1}(i, x') = x\), and the original key is thus recovered.

Updating the sketch amounts to evaluating \(h(x)\) (in time \(\OO(c)\)), and updating the appropriate Split-Sketch in \(\OO(c/\epsilon)\) time.
Decoding the sketch requires the decoding of all Split-Sketches, in total time \(n^{1-\epsilon} \cdot \OO(c/\epsilon \cdot n_2) = \OO(cn/\epsilon)\). Further, all returned keys must be converted through \(h^{-1}\), as described above, taking \(\OO(cn)\) additional time.

We further reduce the space usage by instantiating just a single collection of hash functions, which is shared by all sketches \(T_i\) (rather than instantiating an independent collection for each sketch \(T_i\)).
These hash functions take up \(\tOO(n_2^{1.1} + n_2^{1/3} \log \size{\calU})\) bits of space (by \cref{splitsketch}), which is dominated by the \(\OO(n^{3/4} \log \size{\calU})\) bits needed for the quotient function.

Now define \(\delta \coloneq \sqrt{\log(n)/n^\epsilon} \leq \oo(1)\) and \(n_2 \coloneq (1+\delta)n^\epsilon\).
Let \(\calT\) be the collection of Split-Sketches \(T_i\) that contain at most \(n_2\) keys.
The probability that any of these (at most \(n^{1-\epsilon}\)) sketches fail to decode correctly is bounded by \(\OO(n^{-c})\).
The total error tolerance of these amounts to \(n^{1-\epsilon} \cdot \tOO(n_2^{2/3}) \leq \tOO(n^{1-\epsilon/3})\).

By \cref{quotientfunctions}, with probability \(1 - \OO(n^{-c})\), the number of keys stored in overfull Split-Sketches (those not in \(\calT\)) is bounded by \(\OO(n^{31/32})\).
These keys are all counted towards the error tolerance of the Quotient-Sketch -- this is dominated by the \(\tOO(n^{1-\epsilon/3}) \leq \OO(n^{1-\epsilon/4})\), as \(\epsilon < 1/32\).

Each Split-Sketch \(T_i\) takes up \((1+e^{-\Omega(c/\epsilon)})n_2 (\log_2(2L+1) + \log_2(\size{\calU} / n^{1-\epsilon}))\) bits of space, for a total of
\((1+e^{-\Omega(c/\epsilon)})n (\log_2(2L+1) + \log_2(\size{\calU} / n^{1-\epsilon}))\) bits.

\begin{lemma}[Quotient-Sketch] \label{quotientsketch}
  In the setting of \cref{thm:main} the Quotient-Sketch takes up
  \((1+e^{-\Omega(c/\epsilon)})n (\log_2(2L+1) + \log_2(\size{\calU} / n^{1-\epsilon}))\) bits of space, with error tolerance \(\OO(n^{1-\epsilon/4})\) and failure probability $\OO(n^{-c})$.

  Updates take \(\OO(c/\epsilon)\) time while decoding takes \(\OO(cn/\epsilon)\) expected time.
  The hash functions supplied to the sketch take up \(\OO(n^{3/4} \log \size{\calU})\) bits.
\end{lemma}

\subsection{The Backyard} \label{sec:backyard}
The backyard is a sketch that must decode \emph{all} stored keys with high probability.
As the backyard is meant to retrieve a sublinear number of keys we allow it to use more space per key than the sketches discussed in the preceding sections.
Crucially, however, the backyard must still support \(\OO(c)\) time updates.

Our Backyard-Sketch is based on the IBLT of \cite{goodrich2011invertible} (see \cref{tab:comparison}).
At first glance the IBLT has exactly the properties we need: fast update and decode operations, and it succeeds with high probability.
It deviates from our goal on two counts: It is designed to store key-value pairs, and its analysis assumes access to fully random hashing.
In the following lemma we remedy these two points in the simpler case where we only seek error probability \(\OO(n^{-1})\).

\begin{lemma}[{\cite[Theorem 1]{goodrich2011invertible}}] \label{backyard-inner}
  In the setting of \cref{thm:main} the Patched IBLT is a sketch on \(\OO(n)\) cells which retrieves \(\calS\) with probability \(1 - \OO(n^{-1})\).

  Updates take \(\OO(\epsilon^{-3})\) time, decoding takes \(\OO(n \epsilon^{-3})\) time.
  Additionally, \(\OO(n^{1 + \epsilon} \epsilon^{-3} \log n + \log\size{\calU})\) bits of space are needed to store hash functions and checksums.
\end{lemma}

\begin{proof}
  As discussed above, the Patched IBLT of the lemma diverges from \cite[Theorem 1]{goodrich2011invertible} on two counts: here we sketch a signed multiset rather than key-value pairs, and we claim that the result holds for hash functions constructed from \cref{hashfunctions2}.

  A key \(x\) with multiplicity \(f\) is modeled by storing \(f\) copies of the key-value pair \((x, 1)\) in the Patched IBLT.
  In \cite{goodrich2011invertible} key-value pairs are inserted through an \textsc{Insert} operation (and removed through the symmetric \textsc{Delete} operation).
  In order to increase (or decrease) the multiplicity by \(f\), we could thus naively invoke \textsc{Insert} (or \textsc{Delete}) \(f\) times.
  But as both are linear operations we can easily scale them by \(f\), thus updating the multiplicity of any key in constant time -- the time it takes to insert a single key-value pair.

  Next, we claim that the fully random functions of \cite{goodrich2011invertible} can be replaced by functions constructed from \cref{hashfunctions2}.
  Specifically, each hash function is instantiated with capacity and independence \(\OO(n)\).

  By the same arguments as presented in \cref{sec:decoding-heuristic-sketch}, we evaluate the hash values of at most \(\OO(n)\) distinct keys while decoding the Patched IBLT, and these keys form an Iteratively Defined Set.
  Thus our functions hash all relevant keys independently with probability \(\OO(n^{-1})\) by \cref{hashfunctions2}~(\ref{hashfunctions2:iterativelydefined}), and replacing the fully random functions of \cite{goodrich2011invertible} with our \cref{hashfunctions2} doesn't impact the asymptotic error probability.

  The hash functions take up \(\OO(n^{1 + \epsilon} \epsilon^{-3} \log n + \log\size{\calU})\) bits of space.
  The Patched IBLT additionally maintains a checksum in each cell, using an additional \(\OO(n \log n)\) bits of space across all cells. This is dominated by the size of the hash functions.
\end{proof}

The IBLT specified in \cite{goodrich2011invertible} reduces the error probability to \(n^{-c}\) by increasing the number of locations of each key (and hence the number of hash functions employed) to \(\OO(c)\).
Naturally, we would also have to equip the IBLT with stronger hash functions in order to avoid that the functions fail with probability \(\OO(1/n)\).
But evaluating \(c\) such functions takes \(\OO(c^2)\) time, leading to slow updates.

To boost the error probability to \(\OO(n^{-c})\) we instead instantiate \(c\) copies of the Patched IBLT from \cref{backyard-inner} such that, whp, at least half the sketches are decoded correctly.

\begin{lemma}[The Backyard]\label{thm:backyard}
  In the setting of \cref{thm:main}, and given an \(\epsilon > 0\) the Backyard is a sketch on \(\OO(cn)\) cells which retrieves \(\calS\) with probability \(1 - \OO(n^{-c})\).

  Updates take \(\OO(c \epsilon^{-3})\) time while decoding takes \(\OO(cn \epsilon^{-3})\) time.
  The sketch further needs \(\OO(cn^{1+\epsilon} \epsilon^{-3} \log n + c\log\size{\calU})\) bits of space to store hash functions and checksums.
\end{lemma}
\begin{proof}
  The Backyard consists of \(k=\OO(c)\) independent instances of the Patched IBLT from \cref{backyard-inner}, each of capacity \(n\).
  To add key \(x\), we insert \(x\) into all \(k\) Patched IBLTs, and thus the Backyard performs updates in \(\OO(c \epsilon^{-3})\) time.

  To decode the Backyard, we decode each Patched IBLT separately, obtaining (up to) \(k\) candidate solutions \(S_1, \dots, S_{k}\), each of size \(n\).
  With probability \(\OO(n^{-c})\) less than half of these will be different from \(\calS\).
  Indeed, failures of the Patched IBLTs are indepedent of each other, and thus the probability that \(k/2\) of the sketches fail is bounded by \(\binom{k}{k/2} \cdot \OO(1/n)^{k/2} \leq \OO(n^{-c})\) for \(k=\OO(c)\).

  When \(\calS\) makes up the majority of the candidate solutions, we can identify it through the Boyer-Moore majority vote algorithm in \(\OO(cn)\) time \cite{DBLP:conf/birthday/Moore91}.
\end{proof}

\subsection{The Stuffed IBLT, Proof of Theorem \ref{thm:main}}

\label{sec:maintheorem}
We are finally ready to prove \cref{thm:main}, restated below.
\maintheorem*
\begin{proof}
The Stuffed IBLT contains two structures:
A Quotient-Sketch (\cref{quotientsketch}) with capacity \(n\) and a Backyard with capacity \(n_b\) corresponding to the error tolerance of the Quotient-Sketch. The remaining parameters we fix below, after describing the overall structure.

Each key of \(\calS\) is stored in both structures.
To decode the stored set of keys, we first decode the Quotient-Sketch, which should return a set differing from \(\calS\) in at most \(n_b\) coordinates.
The retrieved keys are subtracted from the backyard, which now contains at most \(n_b\) keys, with multiplicities in $[-2L,2L]$.
These can then be recovered completely, provided we set the multiplicity limit of the backyard to $2L$ (which increases space negligibly compared to multiplicity limit $L$).
Finally, the list of keys recovered from the Backyard is merged with the list from the Quotient-Sketch in \(\OO(n)\) time.
Any key with correct multiplicity $f$ that is erroneously reported with multiplicity \(f'\) when decoding the Quotient-Sketch will be stored in (and thus recovered from) the Backyard with frequency \(f-f'\).
Thus the occurrences cancel out, and the resulting list correctly enumerates \(\calS\).

For the first part of the theorem statement, set the quality of the Quotient-Sketch to \(c\) and the stretch to any small constant \(\epsilon\).
The error tolerance of the Quotient-Sketch is \(\OO(n^{1-\epsilon/4})\), and hence \(n_b = \OO(n^{1 - \epsilon/4})\).

The Quotient-Sketch will decode correctly with probability \(1 - \OO(n^{-c})\). Setting the quality of the Backyard to \(2c\), the failure probability of the Backyard becomes \(\OO(n_b^{-2c}) \leq \OO(n^{-c})\).
Thus the combined sketch recovers \(\calS\) with probability \(1 - \OO(n^{-c})\) as desired.

Finally, we set the stretch of the Backyard to \(\epsilon/8\).
Updates and decodes of the combined sketch takes \(\OO(c)\) and \(\OO(cn)\) time, respectively.
Left is to bound the overall space usage.

The space usage of the Quotient-Sketch is bounded by \((1+e^{-\Omega(c)}) n (\log_2(2L+1) + \log_2(\size{\calU})\) bits, when we discard the space savings due to quotienting.
We now show that the space usage of the remaining elements of the sketch (the Backyard and all required hash functions) take up space sublinear in \(n \log \size{\calU}\).
This makes them lower order terms compared to the bound given above, and their contributions to the space budgets are subsumed by the \(e^{-\Omega(c)}\) term.

First, the hash functions supplied to the Quotient-Sketch take up \(\OO(n^{3/4} \log \size{\calU})\) bits.
The \(\OO(c n_b)\) cells of the backyard likewise take up space sublinear in \(n \log \size{\calU}\), as \(n_b \leq \OO(n^{1 - \epsilon/4})\).
The same goes for its hash functions, as \(n_b^{1+\epsilon/8} \leq n^{1-\epsilon/4}\).
This concludes the analysis of the space usage, and proves the first half of \cref{thm:main}.

For the second part, assume that \(\calU = n^{1+\alpha}\), for a constant \(\alpha > 0\), and let \(\epsilon\) be given.
The sketch works in the same way as described above, except that we set the quality of the Quotient-Sketch to \(c' \coloneq \max\set{c, \OO(\epsilon \ln 1/\epsilon)}\), such that \(e^{-\Omega(c'/\epsilon)} \leq \OO(\epsilon)\).
The remaining parameters we set as described above.

Including the effect of parameter \(\epsilon\), the update time of the Quotient-Sketch is \(\OO(c'/\epsilon) \leq \OO(c/\epsilon^2)\), while that of the Backyard is \(\OO(c/\epsilon^3)\).
The decoding times of the two sketches are bounded by \(\OO(cn/\epsilon^2)\) and \(\OO(cn/\epsilon^3)\), respectively.

We still have to bound the space usage of this sketch in terms of \(\OPT\).
We will show that the space usage of the Quotient-Sketch approximates \(\OPT\) to within a factor of \(1+\OO(\epsilon)\), and afterwards show that the auxiliary space is sublinear in the size of the Quotient-Sketch.

First, we establish a lower bound on \(\OPT\), the minimum space needed to store a multiset of \(n\) distinct keys with multiplicities in \([-L, L]\):
The data structure must store \(n\) values in \([-L, L] \setminus \set{0}\), and encode an \(n\)-subset of \(\calU\).
This necessarily requires at least \(\log_2 \binom{\calU}{n} + \log_2 \left((2L)^n\right)\) bits.
As \(\log_2 \left((2L)^n\right) \geq n \log_2 (2L+1) - n\) and \(\log_2 \binom{\calU}{n} \geq n \log_2 (\size{\calU}/n) = \alpha n \log_2 n\) we obtain the bound \(\OPT \geq n\left(\log_2 (2L+1) + \alpha \log_2 n\right) - n\).

The Quotient-Sketch takes up \((1+\OO(\epsilon))n(\log_2(2L+1) + (\alpha+\epsilon)\log_2(n))\) bits, which is at most
\[
(1+\OO(\epsilon)) \cdot (\OPT+ \epsilon n\log_2(n) + n)) \leq (1+\OO(\epsilon)) \cdot \paren*{1 + \OO(\epsilon/\alpha)} \OPT \leq (1+\OO(\epsilon)) \OPT \, .
\]

Finally, we bound the space usage of the Backyard and the hash functions.
The hash functions supplied to the Quotient-Sketch take up \(\OO(n^{3/4} \log\size{\calU}) = \OO(\alpha n^{3/4} \log n)\), which is sublinear in the space usage of the Quotient-Sketch itself.
For the Backyard and its auxiliary data we have already argued above that the space usage is sublinear in \(n \log\size{\calU} = (1+\alpha) n \log n\), and thus they are likewise dominated by the space usage of the Quotient-Sketch.
This concludes the proof.
\end{proof}

\section{Lower Bounds for Peeling-Based Multiset Sketches}
\label{sec:lower-bound}
In this section, we define the class of \emph{idealised peeling-based multiset sketches} (IPMS) that are granted a \emph{purity oracle} that decides if a given cell is pure.
We derive lower bounds for these hypothetical data structures, implying that no actual sketch using the same setup without an oracle can beat these lower bounds.

Even though our lower bounds do not apply to sketches that combine peeling with something else, we still believe they point to barriers that any approach \emph{primarily} relying on peeling is likely to face.%
\footnote{There are sketches that are peeling-based in the intuitive sense but to which our lower bounds don't apply \cite{baek2023simple,belazzougui2024better}. The issue is that these sketches can accidentally deviate from peeling due to failing peeling heuristics, but recover from such errors eventually. It seems absurd to suggest, however, that such temporary glitches are useful for breaking the lower bounds.}
Given that our own construction reaches these barriers, as shown in \cref{tab:lower-bound-vs-theorem}, we see this as a strong indication that fundamentally new ideas would be needed to surpass it decisively.

Nevertheless, in \cref{sec:beyond-lower-bound} we list ways of stepping outside the box of IPMS that can slightly improve performance.

\begin{definition}
    \label{def:ipms}
    An IPMS for a universe $\U$, a capacity $n ∈ ℕ$ and a multiplicity limit $L ∈ ℕ$ is given by a finite group $G$, a function $π : \U × ℤ → G$ that is injective on%
    \footnote{Injectivity is necessary to decode a pure cell from its content.}
    $\U × (\{-L,…,…L\} \setminus \{0\})$ and linear in its second argument,
    a size $m ∈ ℕ$, and a distribution $\D$ on subsets of $[m]$.

    Each key $x ∈ \U$ is associated with a random set $H_x \sim \D$ of cell indices, independently for each key.%
    \footnote{This is a fitting model if $x$ is mapped to $H_x$ using hashing, but also whenever a (possibly deterministic) sketch is faced with a random input set.}
    The sketch uses an array $T$ of $m$ cells, each storing an element of $G$, initially the neutral element $0 ∈ G$.
    Update and decode are implemented as shown in \cref{alg:ipms}.
\end{definition}
We remark that the sketch behind the main claim of Theorem 1 fits this definition, except for a minor difference with respect to decoding the backyard.\footnote{For the purpose of probability amplification, our backyard uses several sub-sketches that are decoded independently of each other (each effectively using decodeIPMS from \cref{alg:ipms}).}
The “I” in IPBS ensures that the sketch is successfully decoded if and only if the underlying configuration is peelable.

\begin{algorithm}
  \caption{Update and decode procedures of an idealised peeling-based multiset sketch involving a purity oracle. }
  \label{alg:ipms}

  \SetKwFunction{updateIPMS}{updateIPMS}
  \SetKwFunction{decodeIPMS}{decodeIPMS}
  \def\decoded{\textrm{decoded}}
  \proc{\updateIPMS{$x,f$}}{
    \For{$i ∈ H_x$}{
      $T[i] ← T[i] + π(x,f)$\;
    }
  }
  \proc{\decodeIPMS{}}{
    $\decoded ← ∅$\;
    \While(\tcp*[h]{purity oracle}){$∃i ∈ [m]:$ $T[i]$ is pure}{
        let $(x,f) ∈ \U × (\{-L,…,L\} \setminus \{0\})$ such that $π(x,f) = T[i]$ \tcp{unique}
        $\decoded ← \decoded ∪ \set{(x,f)}$\;
        \updateIPMS{$x,-f$}
    }
    \Return $\decoded$\;
  }
\end{algorithm}

\def\err{\mathrm{err}}
\def\E{\mathbb{E}}
\begin{theorem}[Formal version of \cref{thm:lower-bounds-informal}]
    \label{thm:lower-bounds}
    Consider an arbitrary IPMS (given by $\U, G, π, n, L, m, \D$). Let $c = \E_{H \sim \D}[|H|]$ and $\err$ the failure probability when decoding at full capacity%
    \footnote{A sketch is at full capacity if it stores a signed multiset $v ∈ [-L,L]^\U$ with exactly $n$ non-zeroes.}%
    . Then
    \begin{enumerate}[(a)]
        \item The expected update time is $Θ(c)$.
        \item Assuming $\err ≤ \frac 12$ and $\Var{\D} ≤ nc²/16$, the expected decode time is $Ω(cn)$.
        \item Assuming $m ≤ n²$, we have $\err = n^{-\OO(c)}$.
        \item Assuming $\err ≤ \frac 12$ and $c = o(\log n)$ we have $m ≥ (1+e^{-\OO(c)})n$.
        \item The space requirement is at least $m\log₂(|\U|·2L)$ bits.
    \end{enumerate}
\end{theorem}

\begin{proof}
    \begin{enumerate}[(a)]
        \item An update for $x ∈ \U$ involves a loop over $|H_x|$ cells of $T$, hence expected time at least $\E[|H_x|] = c$.
        \item If we had $\err = 0$ then the decode procedure performs exactly $n$ updates and hence takes $Ω(cn)$ expected time. The general argument for $\err ≤ 1/2$ is more complicated because a failing call to decode might be much faster and success might be negatively correlated with the sum of the update costs. The added assumption $\Var{\D} = o(nc²)$ addresses this. Combined with Chebyshev's inequality it implies that the cost of $n$ updates exceeds $nc/2$ with probability at least $\frac{3}{4}$ so that even conditioned on successful decoding the expected cost of the $n$ updates is $Ω(cn)$.
        \item Let $x ≠ y ∈ \U$ be two fixed keys contained in the sketch. Peeling necessarily fails if $H_x = H_y$ as then neither key can be peeled before the other. In the following we bound the probability $p = \Pr[H_x = H_y ∧ |H_x| ≤ 2c]$ of such a full collision in the special case where $|H_x| ≤ 2c$.

        Let $N$ be the number of subsets of $[m]$ of size at most $2c$ and $p₁,…,p_N$ their probabilities under $\D$. We have $p = p₁²+…+p_N²$. By Markov's inequality we know
        \[
            p₁ + … + p_N = \Pr_{H \sim \D}[|H| ≤ 2c] = \Pr_{H \sim \D}[|H| ≤ 2\E[|H|]] ≥ \frac 12.
        \]
        Under this constraint the smallest possible value for $p$ is attained if $p₁ = … = p_N = \frac{1}{2N}$. This implies $p ≥ \big(\frac{1}{2N}\big)²N = \frac{1}{4N}$. Using that $N = \sum_{i = 0}^{2c} \binom{m}{i} ≤ \sum_{i = 0}^{2c}m^i ≤ 2m^{2c} ≤ 2n^{4c}$ we obtain $\err ≥ p ≥ \frac{1}{4N} ≥ \frac{1}{8n^{4c}}$, which implies the claim.
        \item Let $E$ be the (random) number of empty cells at full capacity, i.e.\ cells not used by any key contained in the sketch.
        Successful decoding requires $n+E ≤ m$ because each of the $n$ iterations of decode turns at least one non-empty cell into an empty cell and there are $m$ empty cells in the end. Note that this already implies $m ≥ n$.

        With $p_i := \Pr_{H \sim \D}[i ∈ H]$, the event $E_i$ that the $i$th cell is empty has probability $(1-p_i)^n$ so $\E[E] = \sum_{i ∈ [m]} (1-p_i)^n$. Note that for $H \sim \D$ we have
        \[
        p₁ + … + pₘ = \sum_{i ∈ [m]} \Pr[i ∈ H] = \sum_{i ∈ [m]} \E[\mathbf{1}_{i ∈ H}] = \E[\sum_{i ∈ [m]}\mathbf{1}_{i ∈ H}] = \E[|H|] = c,
        \]
        i.e.\ the average of $p₁,…,p_m$ is $c/m$. Convexity of $x ↦ (1-x)^n$ on $[0,1]$ gives
        \[\sum_{i ∈ [m]} (1-p_i)^n 
        ≥ m·(1-\tfrac{c}{m})^n
        \stackrel{\textrm{(i)}}{≥}
        m·e^{-2cn/m}
        \stackrel{\textrm{(ii)}}{≥}
        n·e^{-2c}
        \]
        where (i) uses that $1-ε ≥ e^{-2ε}$ for $ε ∈ [0,\frac 12]$ and (ii) uses that $m ≥ n$. Hence $\E[E] ≥ n·e^{-2c}$. We assumed $c = o(\log n)$ implying $\E[E] = ω(1)$. Since $E₁,…,E_m$ are negatively associated \cite{JDP83negativeAssociation}\footnote{See \cref{def:NA}, and \cref{lem:NA-applied} for a more detailed discussion of a similar case.}, we can apply a Chernoff bound to show $\Pr[E ≥ \E[E]/2] > \frac 12$. Since $\Pr[E > m-n] ≤ \err ≤ \frac 12$ we have $m-n ≥ \E[E]/2$. Together we get as claimed that $m ≥ n+\E[E]/2 ≥ n+n·e^{-2c} = (1+e^{-2c})n$.
        • The injectivity condition in the definition of IPMSs implies that $|G| ≥ |\U|·2L$. Since we use an array of $m$ elements from $G$ the claim follows.\qedhere
    \end{enumerate}
\end{proof}

\begin{table}[htbp]
    \centering
    \begin{tabular}{rll}
        \toprule
         & lower bound for any IPMS & our sketch\\
         \midrule
         update time & $Θ(c)$ in expectation & $\OO(c)$ in worst case\\
         decode time & $Ω(cn)$ in expectation & $\OO(cn)$ in expectation\\
         failure probability & $n^{-\OO(c)}$ & $n^{-Ω(c)}$\\
         cells used \((= m)\) & $(1+e^{-\OO(c)})n$ & $(1+e^{-Ω(c)})n$\\
         total space & $\#\text{cells}·\log₂(|\U|·2L)$ & $\#\text{cells}·\log₂(|\U|·(2L+1)) + o(n \log \size{\calU})$\\
        \bottomrule
    \end{tabular}
    
    \caption{By \cref{thm:lower-bounds,thm:main}, no idealised peeling based multiset sketch (IPMS, \cref{def:ipms}) can be significantly better than our own construction.
    }
    \label{tab:lower-bound-vs-theorem}
\end{table}

\subsection{Can the Lower Bound be Beaten?}
\label{sec:beyond-lower-bound}
A multiset sketch not strictly adhering to the design described in \cref{def:ipms} might surpass the bounds of \cref{thm:lower-bounds}.

\begin{description}
    •[Error probability.] Most obviously, deterministic sketches achieve $\err = 0$, which no IPMS can \cite{eppstein2010straggler,ganguly2007randomised}.
    •[Space.]
    Using quotienting, as we do in a variant of stuffed IBLTs, it is possible to store less than $\log₂(|\U|·2L)$ bits per cell.
    Viewed abstractly, we drop the condition that a key and its multiplicity can be uniquely decoded from a pure cell's content and replace it with the weaker condition that key and multiplicity can be uniquely decoded from a pure cell's content \emph{and the cell's index}.
    •[Decoding time.]
    An idea in \cite{belazzougui2024better} suggests that decoding can be faster in expectation than $n$ updates, which undermines \cref{thm:lower-bounds} (b).
    Abstractly speaking, we maintain two sketches, a fast one with parameter $c = \OO(1)$ and a slow one with parameter $c' > c$. Updates then take $\OO(c')$ time. A decoding attempt will first consider the fast sketch in isolation, succeeding in time $\OO(n)$ with probability $1-\OO(1/n)$. The slow sketch is only considered if this fails and, in that case, decoded successfully in time $\OO(c'n)$ with probability $1-\OO(n^{-c'})$. The overall expected decoding time is $\OO(n)+\OO(nc'/n) = \OO(n)$ and overall failure probability is $\OO(n^{-c'})$.

    Our lower bound on decode in \cref{thm:lower-bounds} (ii) can be undercut with an implementation of decode that differs from \cref{alg:ipms}.
    
    Note also, however, that the improvement just outlined has to check whether decoding the first sketch has succeeded, which may require additional computations of checksums with total bit complexity $\OO(cn\log n)$. This can be done in $\OO(n)$ time only on word RAMs with word size $w ≥ c \log n$.
\end{description}

These considerations aside, the fact that known deterministic multiset sketches are slow, the fact that all other known multiset sketches rely on peeling, and \cref{tab:lower-bound-vs-theorem}, all taken together, suggest that with our construction we have reached a barrier that is hard to surpass.

\section{Conclusion}
We have seen that it is possible to approach information-theoretical barriers for storing multisets, even for sparse, linear multiset sketches with efficient recovery.
It may still be possible to improve the efficiency of our methods, for example the dependence of the running time on the space overhead. Here are some concrete follow up questions.

\begin{enumerate}

\item Can we support higher multiplicities? The sketches presented in this work only supports limited multiplicity \(L < p/2\), when keys come from universe \(\calU = \mathbb{F}_p\).
Conceivably there are applications where the required multiplicities exceed the universe size.

\item What can be achieved in practice?
Our structures consists of several layers, making it rather impractical. The hidden constants, e.g.\ when we assume ``sufficiently large \(n\)'' are likely very large.

\item What can be done when \(n \approx \calU\)?

\item Can we boost the error probability beyond \(\poly(n)\), without sacrificing space?

\item It would be interesting to avoid the assumption that finite field division can be done in constant time.
Even though it is a very cheap operation in practice, it does take $O(\log p)$ standard word RAM operations to carry out a single division in a finite field.
\end{enumerate}

\iffinal
\paragraph{Acknowledgement.}
We would like to thank Jakob Bæk Tejs Houen for useful discussions in the early stages of this work.
\fi

\bibliography{general.bib}

@book{Dubhashi_Panconesi_2009,
  title={Concentration of Measure for the Analysis of Randomized Algorithms},
  publisher={Cambridge University Press},
  author={Dubhashi, Devdatt P. and Panconesi, Alessandro},
  year={2009}
}

@article{GL2022fusefilter,
  author       = {Thomas Mueller Graf and
                  Daniel Lemire},
  title        = {Binary Fuse Filters: Fast and Smaller Than Xor Filters},
  journal      = {{ACM} J. Exp. Algorithmics},
  volume       = {27},
  pages        = {1.5:1--1.5:15},
  year         = {2022},
  url          = {https://doi.org/10.1145/3510449},
  doi          = {10.1145/3510449},
  bibsource    = {dblp computer science bibliography, https://dblp.org}
}

@article{KRU2011coupling,
  author       = {Shrinivas Kudekar and
                  Thomas J. Richardson and
                  R{\"{u}}diger L. Urbanke},
  title        = {Threshold Saturation via Spatial Coupling: Why Convolutional {LDPC}
                  Ensembles Perform So Well over the {BEC}},
  journal      = {{IEEE} Trans. Inf. Theory},
  volume       = {57},
  number       = {2},
  pages        = {803--834},
  year         = {2011},
  doi          = {10.1109/TIT.2010.2095072}
}

@article{JDP83negativeAssociation,
author = {Kumar Joag-Dev and Frank Proschan},
title = {{Negative Association of Random Variables with Applications}},
volume = {11},
journal = {The Annals of Statistics},
number = {1},
publisher = {Institute of Mathematical Statistics},
pages = {286 -- 295},
year = {1983},
doi = {10.1214/aos/1176346079},
URL = {https://doi.org/10.1214/aos/1176346079}
}

@article{cleary84compact,
  author       = {John G. Cleary},
  title        = {Compact Hash Tables Using Bidirectional Linear Probing},
  journal      = {{IEEE} Trans. Computers},
  volume       = {33},
  number       = {9},
  pages        = {828--834},
  year         = {1984},
  doi          = {10.1109/TC.1984.1676499},
}

@inproceedings{dietz2009splittingtrick,
  author       = {Martin Dietzfelbinger and
                  Michael Rink},
  title        = {Applications of a Splitting Trick},
  booktitle    = {36th {ICALP}},
  series       = {Lecture Notes in Computer Science},
  volume       = {5555},
  pages        = {354--365},
  publisher    = {Springer},
  year         = {2009},
  doi          = {10.1007/978-3-642-02927-1\_30}
}

@article{ganguly2007randomised,
  author       = {Sumit Ganguly},
  title        = {Counting distinct items over update streams},
  journal      = {Theor. Comput. Sci.},
  volume       = {378},
  number       = {3},
  pages        = {211--222},
  year         = {2007},
  url          = {https://doi.org/10.1016/j.tcs.2007.02.031},
  doi          = {10.1016/J.TCS.2007.02.031}
}

@inbook{McDiarmid:1989,
  series={London Mathematical Society Lecture Note Series},
  title={On the method of bounded differences}, DOI={10.1017/CBO9781107359949.008},
  booktitle={Surveys in Combinatorics},
  publisher={Cambridge University Press},
  author={McDiarmid, Colin},
  year={1989},
  pages={148–188},
  collection={London Mathematical Society Lecture Note Series}
}

@inproceedings{ganguly2006deterministic,
  title={Deterministic k-set structure},
  author={Ganguly, Sumit and Majumder, Anirban},
  booktitle={Proceedings of the twenty-fifth ACM SIGMOD-SIGACT-SIGART symposium on Principles of database systems},
  pages={280--289},
  year={2006}
}

@article{eppstein2010straggler,
  title={Straggler identification in round-trip data streams via Newton's identities and invertible {B}loom filters},
  author={Eppstein, David and Goodrich, Michael T},
  journal={IEEE Transactions on Knowledge and Data Engineering},
  volume={23},
  number={2},
  pages={297--306},
  year={2010},
  publisher={IEEE}
}

@inproceedings{goodrich2011invertible,
  title={Invertible {B}loom lookup tables},
  author={Goodrich, Michael T and Mitzenmacher, Michael},
  booktitle={2011 49th Annual Allerton Conference on Communication, Control, and Computing (Allerton)},
  pages={792--799},
  year={2011},
  organization={IEEE}
}

@inproceedings{DBLP:conf/stoc/ChristianiPT15,
  author       = {Tobias Christiani and
                  Rasmus Pagh and
                  Mikkel Thorup},
  editor       = {Rocco A. Servedio and
                  Ronitt Rubinfeld},
  title        = {From Independence to Expansion and Back Again},
  booktitle    = {Proceedings of the Forty-Seventh Annual {ACM} on Symposium on Theory
                  of Computing, {STOC} 2015, Portland, OR, USA, June 14-17, 2015},
  pages        = {813--820},
  publisher    = {{ACM}},
  year         = {2015},
  url          = {https://doi.org/10.1145/2746539.2746620},
  doi          = {10.1145/2746539.2746620},
  bibsource    = {dblp computer science bibliography, https://dblp.org}
}

@InProceedings{fleischhacker2023invertible,
  author =	{Fleischhacker, Nils and Larsen, Kasper Green and Obremski, Maciej and Simkin, Mark},
  title =	{Invertible {B}loom Lookup Tables with Less Memory and Randomness},
  booktitle =	{32nd Annual European Symposium on Algorithms (ESA 2024)},
  pages =	{54:1--54:17},
  series =	{Leibniz International Proceedings in Informatics (LIPIcs)},
  ISBN =	{978-3-95977-338-6},
  ISSN =	{1868-8969},
  year =	{2024},
  volume =	{308},
  editor =	{Chan, Timothy and Fischer, Johannes and Iacono, John and Herman, Grzegorz},
  publisher =	{Schloss Dagstuhl -- Leibniz-Zentrum f{\"u}r Informatik},
  address =	{Dagstuhl, Germany},
  URL =		{https://drops.dagstuhl.de/entities/document/10.4230/LIPIcs.ESA.2024.54},
  URN =		{urn:nbn:de:0030-drops-211252},
  doi =		{10.4230/LIPIcs.ESA.2024.54}
}

@inproceedings{baek2023simple,
  title={Simple set sketching},
  author={B{\ae}k Tejs Houen, Jakob and Pagh, Rasmus and Walzer, Stefan},
  booktitle={Symposium on Simplicity in Algorithms (SOSA)},
  pages={228--241},
  year={2023},
  organization={SIAM}
}

@inproceedings{walzer21,
  author       = {Stefan Walzer},
  editor       = {D{\'{a}}niel Marx},
  title        = {Peeling Close to the Orientability Threshold - Spatial Coupling in
                  Hashing-Based Data Structures},
  booktitle    = {Proceedings of the 2021 {ACM-SIAM} Symposium on Discrete Algorithms,
                  {SODA} 2021, Virtual Conference, January 10 - 13, 2021},
  pages        = {2194--2211},
  publisher    = {{SIAM}},
  year         = {2021},
  url          = {https://doi.org/10.1137/1.9781611976465.131},
  doi          = {10.1137/1.9781611976465.131},
  bibsource    = {dblp computer science bibliography, https://dblp.org}
}

@article{DBLP:journals/mst/FotakisPSS05,
  author       = {Dimitris Fotakis and
                  Rasmus Pagh and
                  Peter Sanders and
                  Paul G. Spirakis},
  title        = {Space Efficient Hash Tables with Worst Case Constant Access Time},
  journal      = {Theory Comput. Syst.},
  volume       = {38},
  number       = {2},
  pages        = {229--248},
  year         = {2005},
  url          = {https://doi.org/10.1007/s00224-004-1195-x},
  doi          = {10.1007/S00224-004-1195-X},
  bibsource    = {dblp computer science bibliography, https://dblp.org}
}

@InProceedings{belazzougui2024better,
  author =	{Belazzougui, Djamal and Kucherov, Gregory and Walzer, Stefan},
  title =	{Better Space-Time-Robustness Trade-Offs for Set Reconciliation},
  booktitle =	{51st International Colloquium on Automata, Languages, and Programming (ICALP 2024)},
  pages =	{20:1--20:19},
  series =	{Leibniz International Proceedings in Informatics (LIPIcs)},
  ISBN =	{978-3-95977-322-5},
  ISSN =	{1868-8969},
  year =	{2024},
  volume =	{297},
  editor =	{Bringmann, Karl and Grohe, Martin and Puppis, Gabriele and Svensson, Ola},
  publisher =	{Schloss Dagstuhl -- Leibniz-Zentrum f{\"u}r Informatik},
  address =	{Dagstuhl, Germany},
  URL =		{https://drops.dagstuhl.de/entities/document/10.4230/LIPIcs.ICALP.2024.20},
  URN =		{urn:nbn:de:0030-drops-201639},
  doi =		{10.4230/LIPIcs.ICALP.2024.20}
}

@article{DBLP:journals/siamdm/SchmidtSS95,
  author       = {Jeanette P. Schmidt and
                  Alan Siegel and
                  Aravind Srinivasan},
  title        = {Chernoff-Hoeffding Bounds for Applications with Limited Independence},
  journal      = {{SIAM} J. Discret. Math.},
  volume       = {8},
  number       = {2},
  pages        = {223--250},
  year         = {1995},
  url          = {https://doi.org/10.1137/S089548019223872X},
  doi          = {10.1137/S089548019223872X},
  bibsource    = {dblp computer science bibliography, https://dblp.org}
}

@inproceedings{DBLP:conf/latin/DemaineHPP06,
  author       = {Erik D. Demaine and
                  Friedhelm Meyer auf der Heide and
                  Rasmus Pagh and
                  Mihai P{\u{a}}tra{\c{s}}cu},
  editor       = {Jos{\'{e}} R. Correa and
                  Alejandro Hevia and
                  Marcos A. Kiwi},
  title        = {De Dictionariis Dynamicis Pauco Spatio Utentibus (\emph{lat.} On Dynamic
                  Dictionaries Using Little Space)},
  booktitle    = {{LATIN} 2006: Theoretical Informatics, 7th Latin American Symposium,
                  Valdivia, Chile, March 20-24, 2006, Proceedings},
  series       = {Lecture Notes in Computer Science},
  pages        = {349--361},
  publisher    = {Springer},
  year         = {2006},
  url          = {https://doi.org/10.1007/11682462\_34},
  doi          = {10.1007/11682462\_34},
  bibsource    = {dblp computer science bibliography, https://dblp.org}
}

@article{gilbert2010sparse,
  author  = {Anna C. Gilbert and Piotr Indyk},
  title   = {Sparse Recovery Using Sparse Matrices},
  journal = {Proceedings of the IEEE},
  volume  = {98},
  number  = {6},
  pages   = {937--947},
  year    = {2010}
}

@inproceedings{DBLP:conf/birthday/Moore91,
  author       = {Robert S. Boyer and
                  J Strother Moore},
  editor       = {Robert S. Boyer},
  title        = {{MJRTY:} {A} Fast Majority Vote Algorithm},
  booktitle    = {Automated Reasoning: Essays in Honor of Woody Bledsoe},
  series       = {Automated Reasoning Series},
  pages        = {105--118},
  publisher    = {Kluwer Academic Publishers},
  year         = {1991},
  bibsource    = {dblp computer science bibliography, https://dblp.org}
}

@inproceedings{yang2024practicalRateless,
author = {Yang, Lei and Gilad, Yossi and Alizadeh, Mohammad},
title = {Practical Rateless Set Reconciliation},
year = {2024},
isbn = {9798400706141},
publisher = {Association for Computing Machinery},
address = {New York, NY, USA},
url = {https://doi.org/10.1145/3651890.3672219},
doi = {10.1145/3651890.3672219},
booktitle = {Proceedings of the ACM SIGCOMM 2024 Conference},
pages = {595–612},
numpages = {18},
location = {Sydney, NSW, Australia},
series = {ACM SIGCOMM '24}
}

@article{eppstein2011setreconciliation,
  title={What's the difference? Efficient set reconciliation without prior context},
  author={Eppstein, David and Goodrich, Michael T and Uyeda, Frank and Varghese, George},
  journal={ACM SIGCOMM Computer Communication Review},
  volume={41},
  number={4},
  pages={218--229},
  year={2011},
  publisher={ACM New York, NY, USA}
}

@article{minsky2003set,
  title={Set reconciliation with nearly optimal communication complexity},
  author={Minsky, Yaron and Trachtenberg, Ari and Zippel, Richard},
  journal={IEEE Transactions on Information Theory},
  volume={49},
  number={9},
  pages={2213--2218},
  year={2003},
  publisher={IEEE}
}

@inproceedings{dodis2004fuzzy,
  title={Fuzzy extractors: How to generate strong keys from biometrics and other noisy data},
  author={Dodis, Yevgeniy and Reyzin, Leonid and Smith, Adam},
  booktitle={International conference on the theory and applications of cryptographic techniques},
  pages={523--540},
  year={2004},
  organization={Springer}
}

@inproceedings{ShibuyaBelazzouguiKucherov2022,
  author = {Yoshihiro Shibuya and Djamal Belazzougui and Gregory Kucherov},
  title = {Efficient Reconciliation of Genomic Datasets of High Similarity},
  booktitle = {WABI},
  year = {2022}
}

@article{DBLP:journals/tcs/BarkayPS15,
  author       = {Neta Barkay and
                  Ely Porat and
                  Bar Shalem},
  title        = {Efficient sampling of non-strict turnstile data streams},
  journal      = {Theor. Comput. Sci.},
  volume       = {590},
  pages        = {106--117},
  year         = {2015},
  url          = {https://doi.org/10.1016/j.tcs.2015.01.026},
  doi          = {10.1016/J.TCS.2015.01.026},
  bibsource    = {dblp computer science bibliography, https://dblp.org}
}

\newpage
\appendix

\section{Adapting a Result Regarding Peeling at High Densities}
\label{sec:importing-spatial-coupling}

In this section, we prove \cref{thm:threshold}, restated here for convenience.
\threshold*

The result is essentially already proved in \cite{walzer21}, except that the relationship between $c$ and the overhead $m/n - 1$ is not explicitly stated and the error probability is claimed to be $o(1)$ rather than $\OO(1/n)$. We deal with both issues separately.
Note that \cite{walzer21} discusses a generalised notion of $ℓ$-peelability while we are only interested in the case $ℓ = 1$.

\paragraph{The relationship between $c$ and the overhead $\frac{m}{n}-1$.}
The main result \cite[Theorem 1.1]{walzer21} refers (using different variable names), for any $c ≥ 3$, to the $c$-orientability threshold $t_c^*$ and guarantees for any $t < t_c^*$ the existence of a distribution $\D$ on sets of size $c$ such that a hypergraph of density $t = \frac{n}{m}$ obtained by repeatedly sampling from $\D$ is peelable with probability $1-o(1)$. It is known, though not mentioned in \cite{walzer21}, that $t_c^* = 1-e^{-Θ(c)}$ \cite[Lemma 1 and Lemma 2]{DBLP:journals/mst/FotakisPSS05}. We apply the Theorem with $t = 2t_c^* - 1$ (such that $t < t_c^*$ and $1-t = 2·(1-t_c^*)$). Note that $t ≥ \frac 12$. Stated as an overhead as in \cref{thm:threshold} (ii) we get $\frac{m}{n}-1 = \frac{1}{t}-1 = \frac{1-t}{t} ≤ 2(1-t) = 4(1-t_c^*) ≤ 4e^{-Θ(c)} = e^{-Θ(c)}$.

\paragraph{The error probability is $\OO(1/n)$.}
In a sense, there is nothing we have to do: The relevant argument from \cite{walzer21} remains valid when changing the definition of “whp” from “with probability $1-o(1)$” to “with probability $1-\OO(1/n)$”. Let us unpack the details.

The relevant claim \cite[Thm 1.2 (i)]{walzer21} regards the probability that a certain hypergraph with $n$ vertices%
\footnote{In this paper $n$ refers to a number of elements and hence hyperedges rather than vertices. Given the linear relationship between the number of vertices and hyperedges, “$\OO(1/n)$” refers to the same asymptotic class in both cases.}
is not peelable. The claim is proved in \cite[Section 4]{walzer21}. Two failure events contribute to the probability.
\begin{itemize}
  • The first failure event is the existence of a small set of edges preventing peelability.
  It is analysed in \cite[Lemma 4.2]{walzer21}, the proof of which shows that such a small obstruction appears with probability at most $\OO(1/n)$.
  • The second failure event relates to the parallel peeling process where, in each round, all vertices of degree at most $1$ are determined and then deleted simultaneously.
  A previous lemma \cite[Lemma 4.1]{walzer21} has demonstrated that for any constant $δ$ there exists $R = R(δ/2)$ such that the \emph{expected} number of vertices remaining after $R$ rounds of parallel peeling is at most $δn/2$. The failure event is that parallel peeling proceeds much slower than expected, with the \emph{actual} number $X_R$ of vertices remaining after $R$ rounds exceeding $δn$.
    The author \cite{walzer21} only gestures towards a corresponding concentration bound with the phrase “standard arguments using Azuma's inequality”. Let us fill in the missing details.
    
    We use that $X_R = f_R(X₁,…,X_N)$ is a function of $N$ independent random variables, with each $X_i$ describing the endpoints of one hyperedge (in our application: the set of positions associated with a key).
    Crucially, each $f_R$ satisfies a bounded difference property: Changing one input can cause the output to drop (or symmetrically to increase) by at most $D = c^R$.
    For $R = 1$ this is simply because changing one hyperedge reduces the degree of at most $c$ vertices. For $R > 1$ we can use induction and the fact that any additional vertex deletion in one round causes at most $c-1$ other vertices to have a smaller degree at the beginning of later rounds.
    
    For such a setting (concretely, $X = f(X₁,…,X_N)$ where $X₁,…,X_N$ are independent and $f$ has $D$-bounded differences) we can use the so-called \emph{method of bounded differences} in the form of McDiarmid's inequality \cite{McDiarmid:1989} (itself derived from Azuma's inequality), which states
    \[
    \text{for } p = \Prp{X-\Ep{X} ≥ t} \text{ we have } p ≤ \exp\Bigg(-\frac{2t²}{ND²}\Bigg) \,.
    \]
    Plugging in $N ≤ n$, $t = δn/2$ and $D = c^R$ we get $p ≤ \exp(-\frac{δ²n}{2c^{2R}})$.
    As \(\delta\) and \(R\) only depend on the overhead (which is a function of \(c\)) we have \(p \leq \OO(1/n)\).
\end{itemize}
By a union bound over the two types of errors, the probability of non-peelability is $\OO(1/n)$. This concludes our argument for \cref{thm:threshold}.

\section{Constructing Small Constant Time Hash Functions} \label{sec:hashing-appendix}
In this appendix we describe the construction of our main hash function, used throughout the paper (\cref{hashfunctions2}, restated below).

\hashfunctionstwo*

The basic building block of our construction is the highly independent functions of \cite{DBLP:conf/stoc/ChristianiPT15}.
\begin{lemma}[{\cite[Corollary 1]{DBLP:conf/stoc/ChristianiPT15}}] \label{lem:basicfunction2}
  For any \(\epsilon < 1\) and positive integers \(n\), \(k\), and \(r\) such that \(k \leq n \leq \size{\calU}\) and \(r \leq n^{\OO(1)}\),
  there exists a hash function \(h \colon [n^3] \to [r]\) such that:
  \begin{enumerate}
  \item The space usage of \(h\) is \(\OO(k n^{\epsilon} \epsilon^{-3} \log n)\) bits.
  \item Constructing and evaluating \(h\) takes time \(\OO(\epsilon^{-3})\).
  \item \(h\) is \(k\)-independent with probability \(1 - n^{-3}\).
  \end{enumerate}
\end{lemma}
\begin{proof}
  \Cref{lem:basicfunction2} follows from Corollary 1 of \cite{DBLP:conf/stoc/ChristianiPT15} by setting \(t = \ceil{3/\epsilon}\) and fixing \(u = n^3\).
  The time bound follows as \(\log n\) (and thus also \(\log k\)) bits fit in a constant number of words.
\end{proof}

The hash functions of \cref{lem:basicfunction2} can be tuned to work on a larger universe (as opposed to the limited range \([n^3]\)).
This comes at the cost of slower execution, however, which we wish to avoid.
Instead, we combine it with a level of universe reduction, as described in the following section.

\subsection{Construction of the Function}
Consider \(\epsilon\), \(m\), \(k\), and \(r\) as in \cref{hashfunctions2}.
Let \(p \colon \calU \to [m^3]\) be a random linear function, and \(g \colon [m^3] \to [r]\) be the hash function of \cref{lem:basicfunction2} with \(n=m\), and \(\epsilon\) and \(k\) coinciding with the given parameters.

Function \(p\) takes up \(\OO(\log \size{\calU})\) bits, and is evaluated in constant time.
Meanwhile, function \(g\) is \(k\)-independent with probability \(1 - m^{-3}\), takes up \(\OO(k m^{\epsilon} \epsilon^{-3}\log m)\) bits, and is constructed and evaluated in \(\OO(\epsilon^{-3})\) time.

The composition \(h = g \circ p\) is then a function \(\calU \to [r]\), and clearly satisfies the time/space guarantees of \cref{hashfunctions2}.
The function partly retains the independence property of \(g\):
For any fixed set \(S\) of size \(\ell \leq m\), \(p(x) \neq p(y)\) for all \(x,y \in S\) with probability at least \(1-\ell^2/m^3\).
Conditioned on the absence of such collisions and on \(g\) being \(k\)-independent, \(h\) is \(k\)-independent on \(S\).
This proves \cref{hashfunctions2}~(\ref{hashfunctions2:independence}).

Extending this property to iteratively defined sets requires a bit more work, we restate the definition here for convenience.
\iterativelydefinedDef*
First, note that the iteratively defined set only observes the ``final'' hash values \(h(x)\).
Had the values \(p(x)\) been observable as well, it would be trivial for an adversary to compute the coefficients of \(p\) and cause collisions.
Intuitively, we thus rely on the \(k\)-independence of \(g\) to ``obscure'' \(p\).

\subsection{Independent Hashing of Iteratively Defined Sets}

We consider a game against any adversary.
The hash function $h = g ∘ p$ is chosen as above (hidden from the adversary). We may condition on the event that $g$ is $k$-independent, which has probability $1-\OO(m^{-3})$.
The adversary may then specify an adaptive sequence of keys $x₁,…,x_k$, learning $h(x_i)$ after specifying $x_i$, and wins the game if and as soon as she sends an $x_i$ such that $p(x_i) = p(x_j)$ for some $1 ≤ j < i ≤ k$. If she loses, then $p(x₁),…,p(x_k)$ are pairwise distinct and hence $h(x₁),…,h(x_k)$ are independent as desired.
We claim that she does not profit from adaptivity and the collision probability with which she wins is the same $\OO(ℓ²/m³)$ as in the offline case.

Indeed, as long as the adversary has not won, the hashes that she observes are fully random and independent of $p$. This means that these observations are not actually useful for winning the game, or more precisely, any adversary that makes use of $h(x₁),…,h(x_{k})$ for generating $x₁,…,xₖ$ can instead make her computation based on internally generated random values without affecting the joint distribution of $(x₁,…,xₖ,p)$. As “winning” is a function of $(x₁,…,xₖ,p)$, this proves our claim.

\section{Quotient Hash Functions}
\label{sec:quotienthashing}
This appendix serves to prove the following version of a result of \cite{DBLP:conf/latin/DemaineHPP06}.
The theorem describes a \emph{quotient} hash function, which can easily be inverted.
The function \(h\) maps each key \(x\) to a pair of values \((h_1(x), h_2(x))\).

\quotientfunctions*

The original presentation of \cite{DBLP:conf/latin/DemaineHPP06} doesn't specify the connection between the running time and the error probability, which we need in our construction.
Further, their presentation doesn't specify the space usage and number of overflowing keys, only that they are sublinear in \(m\).
The construction and proof given in the following section closely mirrors that of \cite{DBLP:conf/latin/DemaineHPP06}, but fills in these omitted details.

\subsection{Construction of the Function}
We will first describe the construction of the function \(h\), and argue for the running time, space usage, and the ability to evaluate \(h^{-1}\).

Throughout the proof we will consider the universe \(\calU\) to be laid out in a grid of \(m^{3/4}\) columns and \(\size{\calU}/m^{3/4}\) rows,
and apply a sequence of transformation to this grid, in order to define the hash value \((h_1(x), h_2(x))\) in the end.
First, we define \(\row(x) \coloneq \floor{x / m^{3/4}}\) and \(\column(x) \coloneq x \bmod m^{3/4}\).

For the first transformation of the grid, we initialize a hash function \(s \colon [\size{\calU}/m^{3/4}] \to m^{3/4}\) as the sum of \(c'\) functions from \cref{hashfunctions2}, each with capacity \(m\) and independence \(m^{1/4}\).
Function \(s\) takes up \(\OO(c' m^{0.3} \log m + c' \log \size{\calU})\) bits and is evaluated in \(\OO(c')\) time.

Each row \(i\) is now shifted by \(s(i)\) positions. We denote the new position of \(x\) in the grid by
\(\column_2(x) = \column(x) + s(\row(x)) \bmod m^{3/4}\).
Note that, given the pair \((\column_2(x), \row(x))\), we can recover \(\column(x)\), and hence \(x\), in \(\OO(c')\) time.

For the next transformation we initialize a collection of \(m^{3/4}\) 2-independent permutations \(\set{\pi^1_i}_{i \in [m^{3/4}]}\), each \(\pi^1_i\) a permutation on \([\size{\calU}/m^{3/4}]\) (the number of rows).
Each 2-independent permutation is a linear map of the form \(x \mapsto ax+b\), and thus takes up \(\OO(\log \size{\calU})\) bits (for a total of \(\OO(m^{3/4} \log \size{\calU})\) bits), and is evaluated in constant time.\footnote{We assume here that $\U / m^{3/4}$ is prime, or rather we make use of a prime slightly larger than $\U / m^{3/4}$, which introduces negligible error terms.}
Note that such a map is easily inverted.

The second transformation permutes the rows within each column by applying the permutation associated with the corresponding column.
We write \(\row_2(x) \coloneq \pi^1_{\column_2(x)}(\row(x))\).
Observe that we can recover \(\row(x)\) from the pair \((\column_2(x), \row_2(x))\).

We now group each set of \(\sqrt{m}\) consecutive columns into a \emph{column group}, for a total of \(m^{1/4}\) column groups. We denote the column group of \(x\) by \(\columng(x) \coloneq \floor{\column_2(x)/\sqrt{m}}\).

Let \(R \coloneq \ceil{\size{\calU} / m^{3/4}}\) be the number of rows. We would like to devide these into $\sqrt{m}$ row groups, with the complication that for $|\U| ≤ m^{5/4}$ this would give more groups than we have rows. In that case we want each row to become its own row group.
Formally, let \(G \coloneq \min\{\sqrt{m}, R\}\) be the number of row groups and \(g \coloneq \ceil{R / G}\) the number of rows in each row group.
We define
\(\rowg(x) \coloneq \floor{\row_2(x) / g}\) and
\(\rowidx(x) \coloneq \row_2(x) \bmod g\).
This makes \(\row_2(x)\) recoverable from \((\rowg(x), \rowidx(x))\).

We initialize a second collection of \(m^{1/4}\) permutations \(\set{\pi^2_i}_{i \in [m^{1/4}]}\), each independently random on the row groups \([G]\).
These can be stored in \(\OO(m^{3/4} \log m)\) bits.

For the final transformation, we permute the order of the row groups within each column group by letting
\(\rowg_2(x) \coloneq \pi^2_{\columng(x)}(\rowg(x))\).

From \((\rowg_2(x), \column_2(x))\) we can recover \(\rowg(x)\).
Because \(\row_2(x)\) is determined by \((\rowg(x), \rowidx(x))\), we can recover \(\row_2(x)\) once we know \(\rowidx(x)\).

In our definition of $(h₁(x),h₂(x))$, the information of $\rowg_2(x)$ will be split between $h₁(x)$ and $h₂(x)$. We let \(\tau \coloneq \floor{G m^{3/4} / b}\) and define
\begin{align*}
  h_1(x) &\coloneq ⟨\column_2(x), \floor*{ \rowg_2(x) / τ }⟩ \\
  h_2(x) &\coloneq ⟨\rowidx(x), \rowg_2(x) \bmod τ⟩.
\end{align*}
where $⟨a,b⟩ = a·c+b$ denotes decoding a pair as an integer, with $c$ being the maximum value that $b$ can attain.

We have \(\column_2(x) < m^{3/4}\) and \(\rowg_2(x) < G\), implying $⌊\rowg_2(x)/τ⌋ < b/m^{3/4}$ and hence \(h_1(x) \in [b]\).
Because \(\rowidx(x) < g = ⌈R/G⌉ = ⌈|\U|/m^{3/4}/G⌉\) and \(\rowg_2(x) \bmod τ < \tau = ⌊Gm^{3/4}/b⌋\), we have \(h_2(x) \in [\size{\calU}/b]\).

Clearly we can recover $(\column_2(x), \rowg_2(x),\rowidx(x))$ from $(h₁(x),h₂(x)$ from which we can recover $x$ as argued above.
All functions are evaluated in \(\OO(c')\) time and take up \(\OO(m^{3/4} \log \size{\calU})\) bits.

\subsection{Bounding the Number of Overflowing Keys}
We call each \(S_i\) a \emph{bucket}, and say that \(S_i\) is \emph{overflowing} if \(\size{S_i} > (1+\delta)m/b\).
In the same way, we say that a key \(x \in S\) is overflowing if \(h_1(x) = i\) and \(S_i\) is an overflowing bucket.
We thus wish to prove a bound on the number of overflowing keys.

We will construct two small sets of \emph{abnormal} keys \(S_1, S_2 \subseteq S\).
We will then analyze the number of overflowing keys, considering only the \emph{normal} keys (ignoring the contribution of abnormal keys), and finally bound the number of additional overflowing keys introduced by adding the abnormal keys.

For each \(i \in [m^{3/4}]\) define \(C_i = \set{x \in S \colon \column_2(x) = i}\) to be the set of keys placed in the \(i\)'th column after the first transformation.
\begin{lemma} \label{transformation1}
  With probability \(1 - \OO(m^{-c'})\), \(\size{C_i} \leq m^{1/4} + m^{3/16}\) for all \(i \in [m^{3/4}]\).
\end{lemma}
\begin{proof}
  Let \(R\) be the indices of rows \(i\) such that there exists an \(x \in S\) with \(\row(x) = i\).
  As \(\size{R} \leq \size{S} = m\), the function \(s\) is \(m^{1/4}\)-independent on these indices with probability at least \(1 - \OO(m^{-c'})\).

  With \(\mu \leq m^{1/4}\) and \(\delta = m^{-1/16}\), it follows from the Chernoff bound  given in \cref{concentration} that \(\size{C_i} \leq m^{1/4} + m^{3/16}\) with probability at least \(\exp(-m^{1/8}/3) \leq 2^{-\Omega(m)}\) for each \(i \in [m^{3/4}]\).
  By a union bound over all \(i\), the statement follows.
\end{proof}

For each column \(C_i\) add \(m^{3/16}\) keys to \(S_1\) (when possible). By \cref{transformation1}, this leaves at most \(m^{1/4}\) normal keys in every column \(C_i\) with high probability.
Meanwhile, \(\size{S_1} \leq m^{3/4} \cdot m^{3/16} = m^{15/16}\).

Let \(S_2 \subseteq S \setminus S_1\) be the set of keys \(x \in S\) such that there exists a different \(y \in S \setminus S_1\) with \(\column_2(x)=\column_2(y)\) and \(\rowg(x)=\rowg(y)\).

\begin{lemma} \label{transformation2}
  With probability \(1 - \OO(m^{-c'})\), \(\size{S_2} \leq 2m^{3/4}\).
\end{lemma}
\begin{proof}
  Note that this claim is trivial in the special case where each row forms its own row group. We may therefore assume that there are $G = \sqrt{m}$ row groups.
  We condition on the distribution of keys into columns and define \(B_i = S_2 \cap C_i\) for each \(i\).
  As each column employs its own permutation for shuffling the rows (and hence computing the row groups) the \(B_i\)'s are independent of each other. Further, \(\size{B_i} \in [0, m^{1/4}]\) when assuming the event of \cref{transformation1}.

  For any fixed keys \(x, y \in C_i\), the probability that \(x\) and \(y\)  will collide in \(\rowg\) is at most \(m^{-1/2}\) (it is slightly lower, as \(\rowg(y)\) is decided by a permutation).
  The probability that \(x\) collides with \emph{any} key \(y \in C_i\) can thus be bounded by \(\size{C_i} \cdot m^{-1/2} \leq m^{-1/4}\), and it follows that the expected size of \(S_2\) is at most \(m^{3/4}\).

  Now consider the scaled variables \(B'_i \coloneq B_i \cdot m^{-1/4}\), whose sum has mean \(\mu \coloneq \Ep{\sum B'_i} = \Ep{\size{S_2}} \cdot m^{-1/4} \leq m^{1/2}\).
  By a standard Chernoff bound, \(\Prp{\sum B'_i > 2m^{1/2}} \leq 2^{-\Omega(m)}\), and we conclude that \(\size{S'} = \sum_i B_i\) will not exceed \(2m^{3/4}\).  
\end{proof}

We are now ready to tackle the normal keys.
\begin{lemma} \label{overflowing1}
  When only considering the normal keys \(S \setminus (S_1 \cup S_2)\), the number of overflowing keys is bounded by \(m \exp(-\delta^2m/(3b)) + m^{15/16}\), with probability \(1 - \OO(m^{-c'})\).
\end{lemma}
\begin{proof}
  We again condition on the distribution of keys into columns and assume the event of \cref{transformation1}.
  Consider two distinct keys \(x, y \in S \setminus (S_1 \cup S_2)\).
  If \(\column_2(x) \neq \column_2(y)\), then \(h_1(x) \neq h_1(y)\).
  Instead assume that \(\column_2(x) = \column_2(y)\), then we must have that \(\rowg(x) \neq \rowg(y)\) -- otherwise \(x\) and \(y\) would have been abnormal keys.
  As each column is split evenly across \(b / m^{3/4}\) buckets, \(\Prp{h_1(x) = h_1(y)} \leq m^{3/4}/b\).
  The permutation \(\pi^2_{\columng(x)}\) is fully independent, and hence the number of keys \(y\) with \(h_1(x) = h_1(y)\) is dominated by a sum of independent variables with mean \(m^{1/4} \cdot m^{3/4}/b = m/b\).

  By a Chernoff bound, the probability that at least \((1+\delta)m/b\) keys are mapped to the same bucket as \(x\) is thus at most \(\exp(-\delta^2 m/(3b))\).
  This is the probability that \(x\) overflows, and it follows that the expected number of overflowing keys is \(\mu \coloneq m\cdot \exp(-\delta^2 m/(3b))\).
  
  Let \(X_i\) be the number of overflowing keys in the \(i\)'th column group.
  As each column group spans \(\sqrt{m}\) columns, \(X_i \leq m^{3/4}\).
  It then follows by a Hoeffding bound that
  \[\Prp{\sum_i X_i > \mu + t} \leq \exp\left(- \frac{2 t^2}{m^{7/4}}\right)\]
  and hence \(\Prp{\sum_i X_i > \mu + m^{15/16}} \leq \exp( - 2m^{1/8}) \leq 2^{-\Omega(m)}\).
  The number of overflowing keys among the \emph{normal} keys is thus bounded by \(m \cdot \exp(-\delta^2 m/(3b)) + m^{15/16}\).
\end{proof}

Finally, we add the abnormal keys back into the process.
Each time an abnormal key \(x\) is assigned to a bucket, one of three things happens:
\begin{enumerate}
\item The bucket is within capacity after the addition of \(x\). The number of overflowing keys thus remains unchanged.
\item The bucket was already overflowing, and \(x\) is the sole new overflowing key.
\item The addition of \(x\) makes a full (but not yet overflowing) bucket overflow. This increases the number of overflowing keys by \((1+\delta) m/b \leq 2m/b\).
\end{enumerate}

Assuming the events of \cref{transformation1,transformation2}, there are at most \(m^{15/16}+2m^{3/4} \leq 3m^{15/16}\) abnormal keys. These lead to at most \(6m^{31/16}/b\) overflowing keys, in addition to those found in \cref{overflowing1}.
This concludes the proof.

\end{document}